%% file: cw-coloring.tex
\documentclass[a4paper,USenglish]{lipics-v2018}
 
\usepackage{microtype}

\title{Finer Tight Bounds for Coloring on Clique-Width}
\titlerunning{Finer Tight Bounds for Coloring on Clique-Width} 

\author{Michael Lampis}{Université Paris-Dauphine, PSL Research University, CNRS, UMR 7243 \\ LAMSADE, 75016, Paris, France}{michail.lampis@dauphine.fr}{}{}
\authorrunning{Michael Lampis} 

\Copyright{Michael Lampis}

\subjclass{Theory of Computation $\rightarrow$ Design and Analysis of Algorithms $\rightarrow$ Parameterized Complexity and Exact Algorithms}

\EventEditors{Ioannis Chatzigiannakis, Christos Kaklamanis, Daniel Marx, and Don Sannella}
\EventNoEds{4}
\EventLongTitle{45th International Colloquium on Automata, Languages, and Programming (ICALP 2018)}
\EventShortTitle{ICALP 2018}
\EventAcronym{ICALP}
\EventYear{2018}
\EventDate{July 9--13, 2018}
\EventLocation{Prague, Czech Republic}
\EventLogo{eatcs}
\SeriesVolume{80}
\ArticleNo{159}

\usepackage{todonotes}
\usepackage{xspace}

\newcommand{\lc}{\textsc{List Coloring}\xspace}
\newcommand{\kc}{\textsc{Coloring}\xspace}
\newcommand{\tsat}{\textsc{3-SAT}\xspace}
\newcommand{\sat}{\textsc{SAT}\xspace} 
\newcommand{\csp}{\textsc{CSP}\xspace}
\newcommand{\qcsp}{$q$-\textsc{CSP}-$B$\xspace}
\newcommand{\mtw}{\ensuremath\mathrm{mtw}}
\newcommand{\mpw}{\ensuremath\mathrm{mpw}}
\newcommand{\pw}{\ensuremath\mathrm{pw}}
\newcommand{\tw}{\ensuremath\mathrm{tw}}
\newcommand{\cw}{\ensuremath\mathrm{cw}} 

\newcounter{count:constr1}

\hideLIPIcs
\nolinenumbers

\begin{document}

\maketitle

\begin{abstract}

We revisit the complexity of the classical $k$-\kc problem parameterized by
	clique-width. This is a very well-studied problem that becomes highly
	intractable when the number of colors $k$ is large. However, much less
	is known on its complexity for small, concrete values of $k$. In this
	paper, we completely determine the complexity of $k$-\kc parameterized
	by clique-width for any fixed $k$, under the SETH. Specifically, we
	show that for all $k\ge 3,\epsilon>0$, $k$-\kc cannot be solved in time
	$O^*\left((2^k-2-\epsilon)^{\cw}\right)$,  and give an algorithm
	running in time $O^*\left((2^k-2)^{\cw}\right)$.  Thus, if the SETH is
	true, $2^k-2$ is the ``correct'' base of the exponent for every $k$.

Along the way, we also consider the complexity of $k$-\kc parameterized by the
	related parameter modular treewidth ($\mtw$).  In this case we show
	that the ``correct'' running time, under the SETH, is
	$O^*\left({k\choose \lfloor k/2\rfloor}^{\mtw}\right)$. If we base our
	results on a weaker assumption (the ETH), they imply that $k$-\kc\
	cannot be solved in time $n^{o(\cw)}$, even on instances with $O(\log
	n)$ colors.

\end{abstract}

\section{Introduction}

\input{intro.tex}

\section{Definitions and Preliminaries}

\input{defs.tex}

\section{SETH and Non-binary CSPs}

\input{csp.tex}

\section{SETH-based Lower Bound for Clique-width}\label{sec:cw}

\input{cw-seth.tex}

\section{Modular Pathwidth and ETH}\label{sec:mtw}

\input{mpw.tex}

\section{Algorithms}\label{sec:algs}

\input{algs.tex}

\section{Conclusions -- Open Problems}

We have given tight bounds for $k$-\kc parameterized by clique-width,
complementing previously known bounds for treewidth. A natural question is now
how robust these bounds are, especially in the context of approximation.
Specifically, does there exist a constant factor approximation algorithm for
$k$-\kc running in $O^*\left((k-\epsilon)^{\tw}\right)$ or
$O^*\left((2^k-2-\epsilon)^{\cw}\right)$? Current knowledge cannot even rule
out the existence of such algorithms with a small \emph{additive} approximation
error and this area is still largely unexplored.

\bibliography{cw-coloring}

\end{document}

%% file: intro.tex
\textsc{Graph Coloring} (from now on simply \kc) is one of the most intensely
studied problems in theoretical computer science. In this classical problem we
are given a graph $G=(V,E)$ and an integer $k$ and are asked if we can
partition $V$ into $k$ sets inducing edge-less graphs. \kc is a notoriously
hard problem as it remains NP-hard even in very restricted cases (e.g.
4-regular planar graphs \cite{Dailey80}) and is essentially completely
inapproximable in general \cite{FeigeK98,Zuckerman07}. This intractability has
motivated the study of the problem in the framework of parameterized
complexity, especially with respect to structural graph parameters.\footnote{In
the remainder, we assume that the reader is familiar with the basics of
parameterized complexity, such as the class FPT, as given in relevant textbooks
\cite{CyganFKLMPPS15,FlumG06}}

Treewidth is by far the most widely studied among such graph parameters, and
\kc has long been known to be FPT by treewidth. This can be seen by either
invoking Courcelle's theorem \cite{0030804}, or by applying a straightforward
dynamic programming technique which, for each bag of a tree decomposition of
width $\tw$ considers all possible $k^{\tw}$ colorings.  Remarkably, thanks to
the work of Lokshtanov, Marx, and Saurabh \cite{LokshtanovMS11a}, we know that
this basic algorithm is likely to be optimal, or at least that improving it
would require a major breakthrough on \sat-solving, as, for any $k\ge 3,
\epsilon>0$, the existence of a $(k-\epsilon)^{\tw}$ algorithm would contradict
the Strong Exponential Time Hypothesis of Impagliazzo and Paturi
\cite{ImpagliazzoP01,ImpagliazzoPZ01}.  More recently, these lower bounds were
strengthened, as Jaffke and Jansen showed that a $(k-\epsilon)^w$ algorithm
would contradict the SETH when $w$ is the graph's vertex edit distance from
being a path \cite{JaffkeJ17}. The same paper showed that the trivial algorithm
can, however, be improved when one considers more restrictive parameters, such
as vertex cover, but still not to the point that the base of the exponent
becomes a constant.  These results thus paint a very clear picture of the
complexity of \kc with respect to treewidth and its restrictions.

One of the drawbacks of treewidth is that it does not cover dense graphs, even
if they have a very simple structure. This has led to the introduction of
clique-width, which is by now (arguably) the second most well-studied
parameter. The complexity of \kc\ parameterized by clique-width has also been
investigated.  Even though \kc is polynomial-time solvable when clique-width is
constant, the best known algorithm for this case runs in time 
$n^{2^{O(\cw)}}$ \cite{KoblerR03}.  Fomin et al.  \cite{FominGLS10} showed that
\kc is not FPT for clique-width (under standard assumptions),  and this was
recently followed up by Golovach et al.  \cite{GolovachL0Z18} who showed,
somewhat devastatingly, that the aforementioned algorithm is likely to be
optimal, as an algorithm running in $n^{2^{o(\cw)}}$ would contradict the ETH.
The problem thus seems to become significantly harder for clique-width, and
this has, in part, motivated the study of alternative dense graph parameters,
such as split-matching width \cite{SaetherT16}, modular-width
\cite{GajarskyLO13}, and twin cover \cite{Ganian11}, all of which make \kc FPT.

\subparagraph*{Contribution:} Although the results mentioned above demonstrate
a clear jump in the complexity of \kc when moving from treewidth to
clique-width, we observe that they leave open a significant hole: all the
aforementioned hardness results for clique-width
(\cite{FominGLS10,GolovachL0Z18}) only apply to the case where $k$ is large
(polynomially related to the size of the graph). It is not hard to see that the
problem becomes significantly easier if both $\cw$ and $k$ are assumed to have
moderate values; indeed \kc is FPT when parameterized by $\cw+k$
\cite{KoblerR03}.  Since the case where $k$ is relatively small is arguably the
most interesting scenario for most applications, we are strongly motivated to
take a closer look at the complexity of \kc parameterized by clique-width, in
order to obtain a more fine-grained and quantitative estimate of the ``price of
generality'' for this problem \emph{for each fixed value of $k$}. Our aim is to
reach tight bounds that paint a crisper picture of the complexity of the
problem than what can be inferred by lower bounds parameterized only by
clique-width, in the same way that the results of \cite{LokshtanovMS11a} do for
$k$-\kc on treewidth.  

The main result of this paper is a lower bound which states that for all $k\ge
3, \epsilon>0$, $k$-\kc cannot be solved in time $O^*\left((2^k-2-\epsilon)^{\cw}\right)$,
unless the SETH is false.  This result gives a concrete, detailed answer to the
question of how much the complexity of $3$-\kc, $4$-\kc, and generally $k$-\kc,
increases as one moves from treewidth to clique-width. As in the lower bound of
\cite{LokshtanovMS11a}, this result is established through a reduction from
\sat.  The main challenge here is that we need to pack a much larger amount of
information per unit of width, and in particular that the graph induced by most
label sets must be edge-less (otherwise many of the $2^k-2$ choices we need to
encode would be invalid).  We work around this difficulty by a delicate use of
the rename operation used in clique-width expressions.

Though having $2^k-2$ in the base of the running time may seem somewhat curious
(and certainly less natural than the $k^{\tw}$ bounds of
\cite{LokshtanovMS11a}), we then go on to prove that this is the ``correct''
bound by giving a matching algorithm. The algorithm is based on standard DP
techniques (including subset convolution \cite{BodlaenderLRV10,RooijBR09}), but
requires a non-standard trick that ``looks ahead'' in the decomposition to
lower the table size from $(2^k-1)^{\cw}$ to $(2^k-2)^{\cw}$. This improves the
previously known DP algorithm of \cite{KoblerR03}, which runs in
$O^*\left(4^{k\cdot\cw}\right)$.

Beyond these results for clique-width we also consider the closely related
parameters modular treewidth and modular pathwidth, which have more recently
been considered as more restricted versions of clique-width
\cite{Mengel16,PaulusmaSS16}. The modular treewidth of a graph $G$ is defined
as the treewidth of the graph obtained from $G$ if one collapses each twin
class into a single vertex, where two vertices are twins if they have the same
neighbors.  By slightly altering our results for clique-width we tightly
characterize the complexity of $k$-\kc for these parameters: the problem is
solvable in time $O^*\left({k\choose \lfloor k/2\rfloor}^{\mtw}\right)$, but not
solvable in $O^*\left(({k\choose \lfloor k/2\rfloor}-\epsilon)^{\mpw}\right)$ under the SETH.
Using the same reduction but relaxing the hypothesis to the ETH, we show that
$k$-\kc cannot be solved in time $n^{o(\mpw)}$, and hence neither in time
$n^{o(\cw)}$ even on instances where $k=O(\log n)$. This can be seen as a
strengthening of the lower bound of \cite{FominGLS10}, which applies only to
clique-width and uses $\Omega(n)$ colors. Our result is incomparable to the
more recent double-exponential bound of \cite{GolovachL0Z18} as it applies to
the more restricted case where the number of colors is logarithmic, and is
tight for this case. Indeed, any reduction giving a double-exponential bound,
such as the one in \cite{GolovachL0Z18}, must inevitably use more than $\log n$
colors, otherwise it would contradict the aforementioned algorithms.

\subparagraph*{Non-binary CSPs.} We mention as a secondary contribution of this
paper a proof that, under the SETH, $n$-variable CSPs over an alphabet of size
$B$ cannot be solved in time $(B-\epsilon)^n$, for any $B,\epsilon$. The
interest of such a result is not so much technical (its proof is implicit in
previous SETH-based bounds, going back to \cite{LokshtanovMS11a}), as
conceptual. Such CSPs provide a convenient starting point for a SETH-based
lower bound for \emph{any base} of the exponential and hence allow us to
isolate a technical part of such proofs from the main part of the reduction.
This explicit statement on the hardness of CPSs has allowed the proofs of this
paper to be significantly shortened, and may facilitate the design of other
SETH-based hardness proofs.

%% file: defs.tex
We use standard graph-theoretic notation and assume that the reader is familiar
with the basics of parameterized complexity, as well as standard notions such
as treewidth \cite{CyganFKLMPPS15,FlumG06}. Let us recall the definition of
clique-width (see \cite{CourcelleO00,CourcelleMR00} for more details). A
labeled graph $G$ has clique-width $w$ if it can be constructed using $w$
labels and the following four basic operations: Introduce($i$), for $i\in
\{1,\ldots,w\}$, which constructs a single-vertex graph whose vertex has label
$i$; Union($G_1,G_2$), which constructs the disjoint union of two labeled
graphs of clique-width $w$; Rename($i,j$) which changes the label of all
vertices labeled $i$ to $j$; and Join($i,j$) which adds all possible edges
between vertices labeled $i$ and vertices labeled $j$. Computing a graph's
clique-width is NP-hard \cite{FellowsRRS09}, and the best currently known
approximation is exponential in clique-width \cite{OumS06}. In this paper, we
will often assume that we are given together with a graph $G$, a clique-width
expression constructing $G$. Since most of our results are negative, this only
makes them stronger, as it shows that our lower bounds do not rely on the
hardness of computing clique-width. We view a clique-width expression as a
rooted binary tree, where the sub-tree rooted in each internal node represents
the corresponding sub-graph of $G$. We use $\cw(G)$ to denote the minimum
number of labels needed to construct a clique-width expression of $G$, and
$\tw(G),\pw(G)$ to denote the treewidth and pathwidth of $G$ respectively.

In a graph $G=(V,E)$ we say that $u,v\in V$ are false twins if $N(u)=N(v)$ and
true twins if $N[u]=N[v]$, where $N[u]=N(u)\cup\{u\}$ denotes the closed
neighborhood of $u$. We say that $u,v$ are twins if they are true or false
twins. We note that in any graph $G$ the partition of vertices into twin
classes is always unique, as the property of being twins is an equivalence
relation \cite{Lampis12}.  Let $G^t$ be the graph obtained from $G$ by deleting
from each twin class all but a single vertex. We define (following
\cite{Mengel16}) the modular treewidth of $G$, denoted $\mtw(G)$, as
$\tw(G^t)$, and similarly the modular pathwidth $\mpw(G)$ as $\pw(G^t)$.

\begin{lemma} \label{lem:parameters}

For all $G$, $\pw(G)\ge \mpw(G) \ge \cw(G)-2$ and $\pw(G)\ge \mpw(G) \ge
\mtw(G)$.

\end{lemma}

\begin{proof}

The inequality $\pw(G)\ge \mpw(G)$ follows trivially from the definition of
modular pathwidth and the fact that deleting vertices from a graph can only
decrease its pathwidth. The inequality $\mpw(G)\ge \mtw(G)$ follows from the
fact that paths are trees. 

Finally, observe that if $u,v$ are twins of $G$ then $\cw(G) = \cw(G-v)$,
because given a clique-width expression for the graph obtained by deleting $v$,
we can obtain a clique-width expression for $G$ by introducing $v$ immediately
after $u$ using a different label, joining them if they are true twins, and
then renaming the label of $v$ to that of $u$. We therefore have $\cw(G) =
\cw(G^t) \le \pw(G^t)+2 =\mpw(G)+2$.  \end{proof}

The Exponential Time Hypothesis (ETH) of Impagliazzo, Paturi, and Zane
\cite{ImpagliazzoPZ01} states that there exists a constant $c_3>1$ such that
\tsat on instances with $n$ variables cannot be solved in time $c_3^n$. If the
ETH is true then we can define, for all $q\ge 3$ a constant $c_q>1$ such that
$q$-\sat, that is, \sat on instances where clauses have maximum size $q$,
cannot be solved in time $c_q^n$. The Strong Exponential Time Hypothesis (SETH)
\cite{ImpagliazzoP01} states that $\lim_{q\to \infty} c_q = 2$, or equivalently
that, for each $\epsilon>0$ there exists a $q$ such that $q$-\sat cannot be
solved in $(2-\epsilon)^n$. We note that sometimes a slightly weaker form of
the SETH is used, which states simply that \sat cannot be solved in
$(2-\epsilon)^n$ for any $\epsilon>0$. The two formulations are not currently
known to be equivalent.  In this paper we use the original, stronger
formulation of \cite{ImpagliazzoP01} (see also e.g. \cite{CyganDLMNOPSW16})
which assumes that $c_q$ tends to $2$.

For any $q,B\ge2$ we define the \qcsp problem as follows: we are given a set
$X$ of $n$ variables which take values in $\{1,\ldots,B\}$, and a set
$\mathcal{C}$ of $q$-constraints on $X$. A $q$-constraint $c$ is defined by an
ordered tuple $V(c)$ of $q$ variables of $X$, and a set
$S(c)\subseteq\{1,\ldots,B\}^q$ of satisfying assignments for $c$.  The
question is whether there exists an assignment $\sigma: X\to \{1,\ldots,B\}$
which satisfies all constraints  $c\in\mathcal{C}$. We say that a constraint
$c\in\mathcal{C}$ is satisfied if applying $\sigma$ to $V(c)$ produces a tuple
of assignments that appears in $S(c)$. To simplify presentation, we will assume
that in the input the list $S(c)$ of the at most $B^q$ satisfying assignments
of each constraint is given explicitly, and that a \qcsp instance is allowed to
contain constraints on fewer than $q$ variables (as we can add dummy variables
to a constraint).

%% file: csp.tex
The SETH states, informally, that as \sat clauses become larger, eventually the
best algorithm for \sat is simply to try out all possible assignments to all
variables. In this section we show that the same is essentially true for \csp
with a larger, non-binary alphabet. The interest in presenting such a result is
that very often we seek to prove a SETH-based lower bound showing that a
problem does not admit an algorithm running in $c^w$, for some constant $c$ and
width parameter $w$ (such as treewidth, or in our case clique-width). This
becomes complicated when we reduce directly from \sat if $c$ is not a power of
$2$ as one cannot make a one-to-one correspondence between binary \sat
variables and ``units of width'' (in our case labels) in the new instance,
which are intended to encode $c$ choices.  As a result, essentially all known
SETH lower bounds of this form include as part of their construction a group
gadget, which maps every $t$ variables of the original \sat instance to $p$
elements of the new problem, for appropriately chosen integers $p,t$ (see e.g.
\cite{BorradaileL16,CyganNPPRW11,JaffkeJ17,LokshtanovMS11a}). Such gadgets are,
however, often cumbersome to design, because they require a problem-specific
trick that expresses a mapping of assignments from a binary to a non-binary
domain.  We therefore prefer to construct a custom-made CSP with a convenient
running time bound, which will later allow us to reduce directly to the problem
we are interested in (\kc on clique-width), in a way that maps exactly one
variable to one clique-width label. This will allow our SETH-based bounds to be
significantly simplified, as we will no longer have to worry about a
discrepancy between the bases of the exponentials.

\begin{theorem}\label{thm:csp}

For any $B\ge2$, $\epsilon>0$ we have the following: if the SETH is true, then
	there exists a $q$ such that $n$-variable \qcsp cannot be solved in
	time $O^*\left((B-\epsilon)^n\right)$.

\end{theorem}

\begin{proof}

Fix a $B,\epsilon$, and suppose that there is an algorithm solving $n$-variable
\qcsp in $(B-\epsilon)^n$ time, for any $q$. We will find some $\delta>0$, such
that we will obtain an algorithm that solves $n$-variable $q$-\sat in time
$(2-\delta)^n$ for any $q$. This will contradict the SETH.

Suppose that we are given a $q$-\sat instance with $n$ variables. We first
define an integer $p$ and a real $\delta>0$ such that there exists an integer
$t$ for which we have:

$$ (B-\epsilon)^p \le (2-\delta)^t < 2^t \le B^p $$

In particular, it suffices to select $p$ so that $B^p\ge 2 (B-\epsilon)^p$ for
such a $t$ to exist. Equivalently, $p\log B \ge 1 + p\log(B-\epsilon)
\Rightarrow p \ge \frac{1}{\log\frac{B}{B-\epsilon}}$, so we set $p = \lceil
\frac{1}{\log\frac{B}{B-\epsilon}}\rceil$, and then a $\delta$ that satisfies
the second inequality always exists.

We now group the $n$ variables of the given \sat instance into $\gamma=\lceil
n/t\rceil$ groups of at most $t$ variables each. Call these groups
$V_1,V_2,\ldots,V_{\gamma}$. For each such group $V_i$ we construct a group of
$p$ variables of the \csp instance, call them $X_i$. Furthermore, we define (in
some arbitrary way) a correspondence which for each assignment to the variables
of $V_i$ gives a distinct assignment to the variables of $X_i$. This is always
possible, since $2^t\le B^p$. Observe that our \csp instance has $\gamma p =
p\lceil n/t \rceil \le \frac{pn}{t}+p$ variables.

Now, for each clause of the $q$-\sat instance, if it involves variables from
the groups $V_{i_1},V_{i_2},\ldots,V_{i_q}$, we construct a constraint that
involves all the variables of the groups $X_{i_1},\ldots,X_{i_q}$ (in other
words, we transform $q$-clauses into $(qp)$-constraints). To define the
satisfying assignments of the constraint, recall that we have defined a
function that maps each assignment of a group $V_i$ to an assignment of a group
$X_i$. We extend these to obtain a function that maps assignments to
$V_{i_1}\cup\ldots\cup V_{i_q}$ to assignments of $X_{i_1}\cup\ldots\cup
X_{i_q}$, and retain as satisfying assignments for the constraint exactly those
assignments that are images of assignments that satisfy the original clause.

This completes the construction, and we remark that everything can be performed
in polynomial time if $q,B,\epsilon$ are constants, as $p,t$ only depend on
$B,\epsilon$. It is also not hard to see that the instances are equivalent: if
there is a satisfying \sat assignment, we give each $X_i$ the assignment that
the satisfying assignment of $V_i$ is mapped to; in the converse direction,
because all constraints are satisfied we have selected for each $X_i$ an
	assignment that is the image of some assignment to $V_i$, and in such a
	way that the obtained assignment for the \sat instance satisfies all
	clauses.

We now invoke the supposed \csp algorithm on our instance where each constraint
involves at most $qp$ variables, and the number of variables is $N=p\lceil
n/t\rceil \le \frac{pn}{t}+p$. The algorithm will run in time 

$$(B-\epsilon)^N \le \big((B-\epsilon)^p\big)^{\frac{n}{t}+1} \le
\big((2-\delta)^t\big)^{\frac{n}{t}+1} \le 2^t\cdot (2-\delta)^n$$

Again, because $t$ only depends on $B,\epsilon$, this running time is
$O^*((2-\delta)^n)$, and we have obtained a faster than $2^n$ algorithm for
$q$-\sat, for any $q$. \end{proof}

%% file: cw-seth.tex
In this section we present our main lower bound result stating that $k$-\kc
cannot be solved in time $O^*\left((2^k-2-\epsilon)^{\cw}\right)$, for any $k\ge 3,
\epsilon>0$, under the SETH. In Section \ref{sec:gadgets} we present some basic gadgets that will also be of use in our lower bound for
modular pathwidth (Section \ref{sec:mtw}). We then present the main
part of the proof in Section \ref{sec:cw-red}.

\subsection{List Coloring and Basic Gadgets}\label{sec:gadgets}

The high-level machinery that we will make use of in our reduction consists of
two major points: first, we would like to be able to express implication
constraints, that is, constraints of the form ``if vertex $u$ received color
$c_1$, then vertex $v$ must receive color $c_2$''; second, we would like to
express disjunction constraints of the form ``at least one of the vertices of
the set $S\subseteq V$ must take a special color $1$''. We build this machinery
in the following steps: first, we show that we can (almost) equivalently
produce an instance of the more general \lc problem; then we use the ability to
construct lists to make \emph{weak edge gadgets}, which for a given pair of
vertices $(u_1,u_2)$ rule out a specific pair of assigned colors; using these
weak edges we construct the aforementioned implication gadgets; and finally we
are able to implement OR constraints using paths on vertices with appropriate
lists.

We give all details for these constructions below. We remark however, for the
convenience of the reader, that a high-level understanding of the informal
meaning of implication gadgets and OR gadgets (precisely stated in Lemmata
\ref{lem:implication-color}, \ref{lem:OR}) is already sufficient to follow the
description of the main part of the reduction, given in Section
\ref{sec:cw-red}. See also Figure \ref{fig:gadgets}.

\subparagraph*{List Coloring} To simplify the presentation of our reduction it
will be convenient to use a slightly more general problem. In \lc, we are given
a graph $G$ and a list of possible colors associated with each vertex and are
asked if there is a proper coloring such that each vertex uses a color from its
list.  This problem clearly generalizes $k$-\kc, as all lists may be
$\{1,\ldots,k\}$.  We will make use of a reduction in the opposite direction.

\begin{lemma}\label{lem:list}

	There is a polynomial-time algorithm which, given an instance of \lc on
	a graph $G$ where all lists are subsets of $\{1,\ldots,k\}$, transforms
	it into an equivalent instance of $k$-\kc on a graph $G'$. Furthermore,
	the algorithm transforms a clique-width expression of $G$ with $\cw$
	labels, to a clique-width expression of $G'$ with $\cw+k$ labels.  If
	all twins of $G$ share the same list, the algorithm transforms a
	modular path decomposition of width $\mpw$ for $G$, to a modular path
	decomposition of $G'$ of width $\mpw+k$.

\end{lemma}

\begin{proof} For the modular pathwidth part, we add to the graph
a clique on $k$ vertices, call them $c_1,\ldots,c_k$. For each vertex $u$ of
$G$ whose list is $L\subseteq \{1,\ldots,k\}$ we connect $u$ to all $c_i$ such
that $i\not\in L$.  It is not hard to see that this produces an equivalent
instance.  Furthermore, if all twins of $G$ share the same list, they remain
twins. We can therefore obtain a modular path decomposition of the new graph by
adding $c_1,\ldots,c_k$ to all bags of the original decomposition. 
	
To obtain a clique-width expression we start with a clique-width expression of
the original graph, and let $l_1,\ldots,l_k$ be $k$ fresh labels. For every
vertex $u$ whose list is $L(u)\subseteq \{1,\ldots,k\}$, we replace its
introduce node with a sub-expression which introduces the vertex $u$, as well
as $k$ new vertices each using a distinct label from $l_1,\ldots,l_k$; and then
performs join operations between the label of $u$ and all labels $l_i$, for
$i\not\in L(u)$. In the root of the expression we add join operations between
every pair of labels in $\{l_1,\ldots,l_k\}$. As a result, we have added to the
graph a complete $k$-partite graph, and each original vertex is connected to a
subset of the parts of this graph in a way that simulates its list.
\end{proof}

\begin{figure}

	\begin{tabular}{ccc} \input{figures/weak-edge.tex} &
	\input{figures/implication.tex}&\input{figures/OR.tex} \end{tabular}
	\caption{Basic gadgets, where empty vertices are internal and solid
	vertices are endpoints that will be connected to the rest of the graph.
	On the left, a weak edge that forbids the combination $(1,5)$ on its
	endpoints. In the middle, an implication that forces color $2$ on the
	right if color $1$ is used on the left. On the right an OR gadget: one
	of the solid vertices must take color $1$.}\label{fig:gadgets}

\end{figure}
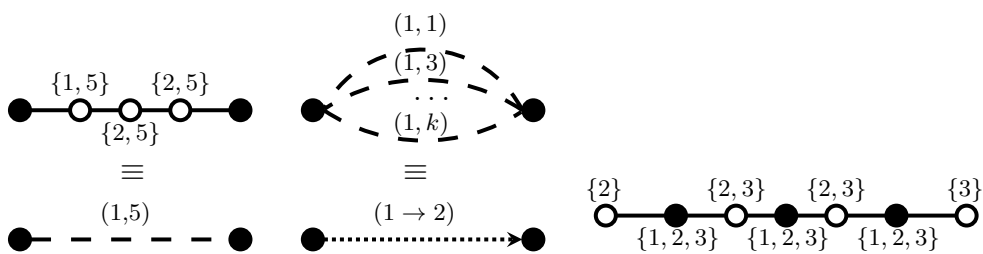

\subparagraph*{Weak Edges and Implications} Normally, the existence of an edge
$(u,v)$ in an instance of \kc forbids the vertices $u,v$ from obtaining the
same color, whatever that color may be.  We will find it convenient to
construct edges that forbid only a specific pair of colors from appearing on 
$u,v$, while allowing any other combination of colors to be used on these two
vertices. Similar versions of this gadget have appeared before, for example
\cite{LokshtanovMS11a,MarxM16}.

\begin{definition}

For two vertices $u_1,u_2$ of a graph $G$ and two colors $c_1,c_2$ a
\emph{$(c_1,c_2)$-weak edge} from $u_1$ to $u_2$ consists of the following:

\begin{enumerate}

\item Three new vertices $v_1,v_2,v_3$ such that $\{u_1,v_1,v_2,v_3,u_2\}$
	induces a path in this order, with endpoints $u_1,u_2$, and
		$v_1,v_2,v_3$ having no edges to the rest of $G$.

\item If $c_1\neq c_2$ let $c'$ be a color distinct from $c_1,c_2$. We assign
to $v_1,v_2,v_3$ the lists $\{c_1,c_2\},\{c_2,c'\},\{c_2,c'\}$.  If $c_1=c_2$,
we assign lists $\{c_1,c'\}, \{c',c''\}, \{c_1,c''\}$ to $v_1,v_2,v_3$
respectively, where $c',c''$ are two distinct colors, different from $c_1$.

\end{enumerate}

\end{definition}


\begin{lemma} \label{lem:weak-color}

Let $G$ be an instance of \lc that contains a $(c_1,c_2)$-weak edge from $u_1$
to $u_2$. Then $G$ does not admit a valid coloring that assigns colors
$(c_1,c_2)$ to $(u_1,u_2)$. Furthermore, if $G'$ denotes the graph obtained by
deleting the internal vertices of the weak edge, any proper list coloring of
$G'$ that does not assign $c_1$ to $u_1$ or does not assign $c_2$ to $u_2$ can
be extended to a proper list coloring of $G$.

\end{lemma}

\begin{proof}

Suppose we assign colors $c_1,c_2$ to $u_1,u_2$. If $c_1\neq c_2$ this would
mean that we would have to color $v_1$ with $c_2$, which forces $v_2$ to color
$c'$, which leaves no available color to $v_3$.  If $c_1=c_2$ and we use this
color on both $u_1,u_2$ then we must give colors $c',c''$ to $v_1,v_3$
respectively, making it impossible to color $v_2$. 

For the converse direction, suppose a coloring of $G'$ assigns a color other
than $c_1$ to $u_1$. We can then give color $c_1$ to $v_1$ and this means that
for any color we give to $v_3$, $v_2$ has a color available to complete the
coloring. If on the other hand a coloring gives a color other than $c_2$ to
$u_2$, we give color $c_2$ to $v_3$. If $c_1=c_2$ we are done, because the list
of $v_2$ does not contain $c_2$, so we can always extend a coloring of $v_1$ to
$v_2$. If $c_1\neq c_2$, we color $v_2$ with $c'$, and this will not contradict
any coloring of $v_1$, since $c'$ is not in the list of $v_1$.  \end{proof}

Let us now use the weak edges we have defined above to construct an implication
gadget. The intuitive meaning of placing an implication gadget from a vertex
$u_1$ to a vertex $u_2$ is to impose the constraint that if $u_1$ is assigned
color $c_1$, then $u_2$ must be assigned color $c_2$.

\begin{definition}

For two vertices $u_1,u_2$ and two colors $(c_1,c_2)$ a $(c_1\to
c_2)$-implication from vertex $u_1$ to vertex $u_2$ is constructed as follows:
for each color $c'\neq c_2$, we add a $(c_1,c')$-weak edge from $u_1$ to $u_2$.

\end{definition}

\begin{lemma} \label{lem:implication-color}

Let $G$ be an instance of \lc that contains a $(c_1\to c_2)$-implication from
$u_1$ to $u_2$. Then $G$ does not admit a list coloring that gives color $c_1$
to $u_1$ and a color $c'\neq c_2$ to $u_2$. Furthermore, if $G'$ is the graph
obtained from $G$ by deleting the internal vertices of the implication gadget,
any coloring of $G'$ that either does not assign $c_1$ to $u_1$, or assigns
$c_2$ to $u_2$ can be extended to a coloring of $G$.

\end{lemma}

\begin{proof}

Follows directly from the proof of Lemma \ref{lem:weak-color}. In particular,
if $u_1$ receives $c_1$ and $u_2$ receives $c'\neq c_2$, by construction there
exists a $(c_1,c')$-weak edge from $u_1$ to $u_2$ which cannot be colored. For
the converse direction, all weak edges of the implication gadget are activated
either by setting $u_1$ to $c_1$, or setting $u_2$ to a color other than $c_2$.
\end{proof}

\begin{lemma} \label{lem:weak-pw}

Let $G$ be an instance of \lc, and $G'$ be the graph obtained from $G$ by
replacing every $(c_1,c_2)$-weak edge or $(c_1\to c_2)$-implication gadget with
endpoints $u_1,u_2$ with an edge $(u_1,u_2)$ (or simply deleting the internal
vertices of the weak edge if $(u_1,u_2)$ already exists).  Then $\pw(G) \le
\pw(G') + 3$.

\end{lemma}

\begin{proof}

We prove the lemma just for $(c_1,c_2)$-weak edges, since implication gadgets
are just collections of weak edges that share the same endpoints.

Consider a path decomposition of $G'$. We construct a path decomposition that
contains the internal vertices of a weak edge with endpoints $u_1,u_2$ as
follows: first we find a bag of the decomposition of $G'$ that contains both
$u_1,u_2$ (such a bag exists because $(u_1,u_2)$ is an edge of $G'$); then we
insert after this bag an identical bag, into which we insert the three internal
vertices of the weak edge. We repeat this process for all weak edges.
\end{proof}

\subparagraph*{OR gadgets} We will also make use of a gadget that forces any
valid list coloring of a graph to assign a special color $1$ to one vertex out
of a set of vertices. Invariably, the idea will be that this will be a color
that activates some implications, allowing us to propagate information about
the coloring between parts of the graph. We recall that a similar version of an
OR gadget was also used in \cite{LokshtanovMS11a}.

\begin{definition}

	An OR gadget on an independent set of vertices $S$, denoted OR($S$), is
	constructed as follows: we assign list $\{1,2,3\}$ to all vertices of
	$S$; we construct a new set $S'$ of internal vertices (that will not be
	connected to the rest of $G$), such that $|S'|=|S|+1$ and $S\cup S'$
	induces a path alternating between vertices of $S$ and $S'$; we assign
	list $\{2,3\}$ to all vertices of $S'$, except the two endpoints of the
	path, which receive lists $\{2\},\{3\}$, respectively.

\end{definition}

\begin{lemma} \label{lem:OR}

	If $G$ is a \lc instance that contains an OR($S$) gadget then $G$ does
	not admit a list coloring that does not use color $1$ in any vertex of
	$S$. Furthermore, for any vertex $u\in S$, there exists a proper list
	coloring of the graph induced by the gadget that assigns color $1$ only
	to $u$.

\end{lemma}

\begin{proof}

Observe that a coloring that does not use color $1$ in $S$ would be a
	two-coloring of the path induced by $S\cup S'$ using colors $\{2,3\}$.
	However, this path has an odd number of vertices, therefore the
endpoints would need to be assigned the same color. Because we have assigned to
the endpoints the singleton lists $\{2\},\{3\}$, this is impossible. For the
second claim, if we give color $1$ to a vertex $u\in S$, the remaining vertices
induce two disjoint paths, each of which has a single vertex with a singleton
list, hence both paths can be two-colored with $\{2,3\}$. \end{proof}

\subsection{Reduction for Clique-width}\label{sec:cw-red}


\begin{theorem}\label{thm:cw}

For any $k\ge3, \epsilon>0$, if there exists an algorithm solving $k$-coloring
	in time $O^*\left((2^k-2-\epsilon)^{\cw}\right)$, where $\cw$ is the
	input graph's clique-width, then the SETH is false.

\end{theorem}

The proof of Theorem \ref{thm:cw} consists of a reduction from a \csp produced
by Theorem \ref{thm:csp}. Before giving details, let us give some intuition.
Our new instance will be a graph and a clique-width expression with, roughly,
$n$ labels, where $n$ is the number of variables of the \csp instance. The set
of colors used in each label will encode the value given to a variable in a
satisfying assignment. As a result, with $k$ colors, we will have $2^k-2$
encodings available, as every label set uses at least one color, but will never
use all $k$ colors. To verify that these assignments are correct, we will
construct for each constraint an OR gadget which forces the use of color $1$ on
a vertex representing a particular assignment. This assignment dictates the
value of each variable of the constraint, and therefore the set of colors used
in some of our label sets. To verify that the assignment is consistent we use
implication gadgets that force some auxilliary vertices to receive the
complement of the colors dictated by the constraint assignment, and then
connect these with the vertices encoding the true assignment. If the assignment
used is truly consistent, these edges will end up being properly colored.

\begin{figure}

\begin{tabular}{cc}

\input{figures/red1.tex} & \input{figures/red2.tex}

\end{tabular}

\caption{Left: high-level view of the reduction. Rows correspond to variables,
columns to constraints. Here, variable $x_1$ appears in constraints
$c_0,c_1,c_3,c_4$. Right: connections between the OR gadgets OR$(S_j)$  and the
$V_i^j,U_i^j$ sets.  Giving color $1$ to $a_1$ represents selecting this
assignment. This forces the use of some colors in $V_1^{0,a_1}$ and the
complementary set in $U_1^{0,a_1}$.}

\end{figure}
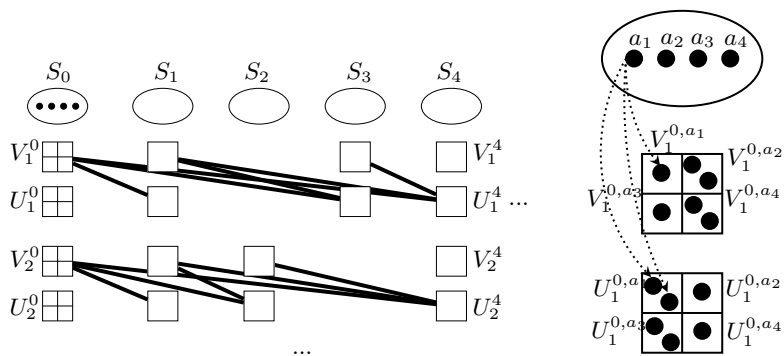

\subparagraph*{Construction.} We are given $k\ge3, \epsilon>0$. Let
$B=2^k-2$. Let $q$ be the smallest integer such that $n$-variable \qcsp does
not admit an $O^*\left((B-\epsilon)^n\right)$ algorithm.  According to Theorem
\ref{thm:csp}, $q$ exists if the SETH is true, and it depends only
on $B,\epsilon$.  Consider an arbitrary $n$-variable instance of \qcsp, call it
$\phi$.  We use the existence of the supposed $O^*\left((2^k-2-\epsilon)^{\cw}\right)$
algorithm to obtain an $O^*\left((B-\epsilon)^n\right)$ algorithm that decides
$\phi$, contradicting the SETH.

We define in some arbitrary way a translation function $T$ which, given a value
$v\in\{1,\ldots,B\}$ returns a non-empty proper subset of $\{1,\ldots,k\}$. We
make sure that $T$ is defined in such a way that it is one-to-one; this is
possible since the number of non-empty proper subsets of $\{1,\ldots,k\}$ is
exactly $B=2^k-2$.

Let $X=\{x_1,\ldots,x_n\}$ be the set of the $n$ variables of the \qcsp
instance and $C=\{c_0,\ldots,c_{m-1}\}$ the set of its $m$ constraints. Let
$L=3m(nk+1)$. We now construct our graph, where if we don't specify the list of
a vertex it can be assumed to be $\{1,\ldots,k\}$. For each
$j\in\{0,\ldots,L-1\}$ we do the following:

\begin{enumerate}

	\item \label{it:s11} Let $j'= j \bmod m$ and let $S$ be the set of
		satisfying assignments of the constraint $c_{j'}$. We construct
		an independent set of vertices $S_j$ that contains a vertex for
		every assignment of $S$. We construct an OR($S_j$) gadget on
		these vertices.

	\item \label{it:s12} For each $x_i$ which appears in $c_{j'}$ and for
		each assignment $a\in S$ we do the following:

			\begin{enumerate}

				\item Let $v\in\{1,\ldots,B\}$ be the value
					given to $x_i$ by the assignment $a$.
					Construct an independent set
					$V_i^{j,a}$  of $|T(v)|$ vertices and
					an independent set $U_i^{j,a}$ of
					$k-|T(v)|$ vertices.  Recall that
					$T(v)$, the translation function,
					returns a set of size between $1$ and
					$k-1$, so both these sets are
					non-empty.

				\item\label{it:s2} For each color $c\in T(v)$
					select a distinct vertex in $V_i^{j,a}$
					and add a $(1 \to c)$-implication
					gadget from the vertex that represents
					the assignment $a$ in $S_{j}$ to this
					vertex of $V_i^{j,a}$.

				\item\label{it:s3} For each color $c\in
					\{1,\ldots,k\}\setminus T(v)$ select a
					distinct vertex in $U_i^{j,a}$ and add
					a $(1 \to c)$-implication gadget from
					the vertex that represents the
					assignment $a$ in $S_{j}$ to this
					vertex of $U_i^{j,a}$.

				\item\label{it:s4} Connect all vertices of
					$U_i^{j,a}$ with all vertices of
					previously constructed sets
					$V_i^{l,a'}$, for all $l< j$ and all
					assignments $a'$. 

			\end{enumerate}

\end{enumerate}

This completes the construction, and we call the constructed \lc instance
$G(\phi)$. The intended meaning is that the sets $V_i^{j,a}$ will use a set of
colors that encodes the value of the variable $x_i$, while the sets $U_i^{j,a}$
will use colors from the complement of this set.

\begin{lemma}\label{lem:cw1}

	If $\phi$ is a satisfiable \qcsp instance, then $G(\phi)$ admits a
	proper list coloring.

\end{lemma}

\begin{proof}

Suppose that we have a satisfying assignment for $\phi$ which gives value $v_i$
	to variable $x_i$. The invariant we will maintain is that for all
	$j,a$, all vertices of sets $V_i^{j,a}$ will use only colors from
	$T(v_i)$, while all vertices of sets $U_i^{j,a}$ will use only colors
	from $\{1,\ldots,k\}\setminus T(v_i)$. As a result, all edges added in
	step \ref{it:s4} will be properly colored. The rest of the graph will
	be easy to color if we respect the informal meaning of OR and
	implication gadgets.
	
	More specifically, for each OR($S_j$) gadget we let $j'=j\bmod m$ and
	consider the constraint $c_{j'}$. The supposed assignment to $\phi$
	assigns to the variables of the constraint values consistent with a
	satisfying assignment $a$ of $c_{j'}$. We give color $1$ to the
	corresponding vertex of $S_{j}$. We use colors $\{2,3\}$ to color all
	remaining vertices of the OR gadget. Note that the OR gadget is
	connected to the rest of the graph only through implication gadgets
	activated by color $1$. Hence, by Lemma \ref{lem:implication-color} we
	can remove all non-activated implication gadgets.  For the remaining,
	activated implication gadgets we color their other endpoints, which are
	found in the sets $V_i^{j,a}$ and $U_i^{j,a}$ with the unique viable
	color. For every other assignment $a'\neq a$ we color all vertices of
	$V_i^{j,a'}$ using a color we used in $V_i^{j,a}$, and the vertices of
	$U_i^{j,a'}$ using a color we used in $U_i^{j,a}$. 
	
	The promised invariant is maintained, as the vertices of $V_i^{j,a}$
	are forced to receive colors from $T(v_i)$, while vertices of
	$U_i^{j,a}$ are forced to receive colors from the complementary set.
	Thus, all edges of step \ref{it:s4} are properly colored, and since we
also properly colored the OR gadgets and implication gadgets, we have a proper
coloring of the whole graph.  \end{proof}

\begin{lemma}\label{lem:cw2}

	If  $G(\phi)$ admits a proper list coloring, then $\phi$ is a
	satisfiable \qcsp instance.

\end{lemma}

\begin{proof} Suppose we have a list coloring of $G(\phi)$ given by the
function $\mathbf{c}:V\to\{1,\ldots,k\}$.  For a set $V'\subseteq V$ we will
write $\mathbf{c}(V')$ to denote the set of colors used by $\mathbf{c}$ for
vertices of $V'$, that is, $\mathbf{c}(V') = \{ c\ |\ \exists u\in V',\
\mathbf{c}(u)=c\ \}$.  Let $j\in\{0,\ldots,L-1\}$, $j' = j\bmod m$, and $S$ be
the set of satisfying assignment of the constraint $c_{j'}$, which contains a
variable $x_i$. Consider the set $V_i^j = \cup_{a\in S} V_i^{j,a}$.  We define
the \emph{candidate assignment} of $x_i$ at index $j$ as
$v_i^j:=T^{-1}(\mathbf{c}(V_i^j))$. In other words, to obtain the candidate
assignment for $x_i$ at index $j$, we take the union of all colors used in
$V_i^{j,a}$, and then translate this set back into a value in $\{1,\ldots,B\}$. 
	
	We observe that for all $i,j$, such that $x_i$ appears in $c_{j'}$,
	where $j'=j\bmod m$, there exists an assignment $a$ such that
	$\mathbf{c}(V_i^{j,a}) = \{1,\ldots,k\}\setminus
	\mathbf{c}(U_i^{j,a})$.  To see this, note that by Lemma \ref{lem:OR},
	one of the vertices of the OR($S_j$) gadget must have received color
	$1$, say the vertex that corresponds to assignment $a$.  All the
	implications incident on this vertex are therefore activated, which
	means that, if $a$ gives value $v\in\{1,\ldots,B\}$ to $x_i$, then
	$\mathbf{c}(V_i^{j,a})=T(v)$ and
	$\mathbf{c}(U_i^{j,a})=\{1,\ldots,k\}\setminus T(v)$ (because of the
	implications of steps \ref{it:s2},\ref{it:s3} and Lemma
	\ref{lem:implication-color}). 
	
	A key observation now is the following: for all $j_2>j_1$ and for all
	$i$ such that variable $x_i$ appears in constraints $c_{j_1'},c_{j_2'}$
	with $j_1' = j_1\bmod m$, $j_2'=j_2\bmod m$, we have
	$\mathbf{c}(V_i^{j_1}) \subseteq \mathbf{c}(V_i^{j_2})$.  In other
	words, the set of colors used in $V_i^j$ can only increase as $j$
	increases. To see this, suppose that there exists $c\in
	\mathbf{c}(V_i^{j_1})\setminus \mathbf{c}(V_i^{j_2})$. As argued in the
	previous paragraph, there exists an assignment $a_2$ such that
	$\mathbf{c}(V_i^{j_2}) \supseteq \mathbf{c}(V_i^{j_2,a_2})=
	\{1,\ldots,k\}\setminus \mathbf{c}(U_i^{j_2,a_2})$.  Because $c\not\in
	\mathbf{c}(V_i^{j_2,a_2})$ we must have $c\in
	\mathbf{c}(U_i^{j_2,a_2})$, but because of step \ref{it:s4}, all
	vertices of $U_i^{j_2,a_2}$ are connected to all of $V_i^{j_1}$. Since
	$c\in\mathbf{c}(V_i^{j_1})$, this contradicts the correctness of the
	coloring.

	The property established in the previous paragraph implies that for
	each $i$ there exist at most $k$ distinct candidate assignments $v_i^j$
we can obtain for different values of $j$, as each assignment is obtained by
translating the set of colors used in $V_i^j$, this set only increases, it
always contains at least one color and at most $k$ colors.  Let us say that an
index $j_1$ is problematic if, for some $i\in\{1,\ldots,n\}$ we have the
following: $x_i$ appears in constraint $c_{j_1'}$, where $j_1' = j_1\bmod m$;
and if $j_2$ is the minimum index such that $j_2>j_1$ and $x_i$ appears in
constraint $c_{j_2'}$, where $j_2'=j_2\bmod m$, then $v_i^{j_1}\neq v_i^{j_2}$.
In other words, an index is problematic if the candidate assignment it produces
for a variable disagrees with the candidate assignment produced for the same
variable in the next index that involves this variable. It is not hard to see
that there are at most $kn$ problematic indices, because for each variable
there are at most $k$ problematic indices.  Therefore, since $L = 3m(nk+1)$, by
pigeonhole principle, there exists an interval $L'$ of at least $3m$
consecutive non-problematic indices.

	We now obtain an assignment for the original instance as follows: for
	each variable $i$ we take an index  $j\in L'$ such that $x_i$ appears
	in constraint $c_{j'}$, where $j'=j\bmod m$,  and give $x_i$ the
	candidate value $v_i^j$ from that index.  Observe that, by the
	definition of $L'$ the index we select is irrelevant, as all candidate
	values are constant throughout the interval $L'$. 
	
	We claim that this is a satisfying assignment.  Suppose not, so there
exists an unsatisfied constraint $c_{j'}$.  Because $L'$ contains $3m$
consecutive indices, there exists three indices $j_1<j_2<j_3\in L'$ such that
$j' = j_1 \bmod m = j_2\bmod m = j_3\bmod m$.  We observe that for all
variables $x_i$ appearing in $c_{j'}$ we have given value $v_i^{j_2}$, that is
the candidate value obtained at index $j_2$, since all indices in $L'$ give the
same candidate values to all variables. 
	
	Now, there exists a vertex in $S_{j_2}$ that received color $1$,
	representing an assignment $a$.  If the assignment we produced is not
	consistent with $a$, there exists a variable $x_i$ such that we have
	given $x_i$ value $v=v_i^{j_2}$, while $a$ gives it value $v'\neq v$.
	Consider now the set $V_i^{j_2,a}$.  Because of the implication
	gadgets, it uses the colors $T(v')\neq T(v)$. If there exists $c\in
	T(v)\setminus T(v')$ then $c\in \mathbf{c}(U_i^{j_2,a})$.  But
	$U_i^{j_2,a}$ is connected to all vertices of $V_i^{j_1}$, which, we
	assumed use all colors of $T(v)$, therefore also color $c$,
	contradicting the correctness of the coloring. If on the other hand
	there exists $c\in T(v')\setminus T(v)$, then since
	$\mathbf{c}(V_i^{j_3}) = T(v)$, there exists an $a'$ such that
	$\mathbf{c}(U_i^{j_3,a'}) = \{1,\ldots,k\}\setminus T(v)$. Therefore,
	$c\in \mathbf{c}(U_i^{j_3,a'})$, while $c\in \mathbf{c}(V_i^{j_2,a})$,
and by step \ref{it:s4} these sets are connected, again obtaining a
contradiction.  We therefore conclude that we must have a consistent satisfying
assignment.  \end{proof}

\begin{lemma}\label{lem:cw3}

	$G(\phi)$ can be constructed in time polynomial in $|\phi|$, and we
have $\cw(G(\phi)) \le n + O(qk^2B^q) = n + f(\epsilon,k)$ for some function $f$.

\end{lemma}

\begin{proof}

For fixed $k\ge 3,\epsilon>0$, we have that $B=2^k-2$ and $q$ is a constant
	that only depends on $B,\epsilon$ (that is, on $k,\epsilon$). Each
	constraint of the \qcsp instance has at most $B^q$ satisfying
	assignments. Therefore, it is not hard to see that the whole
	construction can be performed in polynomial time, if $k,\epsilon,B,q$
are constants.  For clique-width we use the following labels:

	\begin{enumerate}

		\item $n$ main labels, representing the variables of $\phi$.

		\item A single junk label. Its informal meaning is that a
			vertex that receives this label will not be connected
			to anything else not yet introduced in the graph.

		\item $O(B^q)$ constraint work labels.

		\item $O(qk^2B^q)$ variable-constraint incidence work labels.

	\end{enumerate}

	To give a clique-width expression we will describe how to build the
	graph, following essentially the steps given in the description of the
	construction by maintaining the following invariant: before starting
	iteration $j$, all vertices of the set $\bigcup_{j'<j}V_i^{j',a}$
	(where we take the union over all assignments $a$), have label $i$, and
	all other vertices have the junk label.

	This invariant is vacuously satisfied before the first iteration, since
	the graph is empty. Suppose that for some $j\in\{0,\ldots,L-1\}$ the
	invariant is true. We use the $O(B^q)$ constraint work labels to
	introduce the vertices of the OR($S_j$) gadget of step \ref{it:s11},
	giving each vertex a distinct label. We use join operations to
	construct the internal edges of the OR gadget. 
	
	Then, for each variable $x_i$ that appears in the current constraint we
	do the following: we use $O(k^2B^q)$ of the variable-constraint
	incidence work labels to introduce the vertices of $V_i^{j,a}$,
	$U_i^{j,a}$ as well as the implication gadgets connecting these to
	$S_j$. Again we use a distinct label for each vertex, but the number of
	vertices (including internal vertices of the implication gadgets) is
	$O(k^2B^q)$, so we have sufficiently many labels to use distinct labels
	for each of the $q$ variables of the constraint.  We use join
	operations to add the edges inside all implication gadgets.  Then we
	use join operations to connect $U_i^{j,a}$ to all vertices
	$\bigcup_{j'<j}V_i^{j',a}$, for $j'<j$.  This is possible, since the
	invariant states that all the vertices of $\bigcup\limits_{j'<j
	}V_i^{j',a}$ have the same label.  We then rename all the vertices of
	$U_i^{j,a}$, for all $a$ to the junk label, and do the same also for
	internal vertices of all implication gadgets.   We proceed to the next
	variable of the same constraint and handle it using its own $O(k^2B^q)$
	labels.  Once we have handled all variables of the current constraint,
	we rename all vertices of each $V_i^{j,a}$ to label $i$ for all $a$.
	We then rename all vertices of the OR($S_j$) gadget to the junk label
	and increase $j$ by $1$. It is not hard to see that we have maintained
the invariant and constructed all edges induced by the vertices introduced in
steps up to $j$, so repeating this process constructs the graph.  \end{proof}

\begin{proof}[Theorem \ref{thm:cw}]

	The proof now follows from the described construction, Theorem
	\ref{thm:csp}, and Lemmata \ref{lem:list}, \ref{lem:cw1},
	\ref{lem:cw2}, \ref{lem:cw3}. In particular, given any fixed $k\ge3,
	\epsilon>0$, our construction, in combination with Lemma \ref{lem:list}
	and \ref{lem:cw1} produces a $k$-\kc instance of clique-width
	$n+f(k,\epsilon)$ and size polynomial in the original \qcsp instance.
	If there exists an algorithm running in
	$O^*\left((2^k-2-\epsilon)^{\cw}\right)$ for $k$-\kc, since we have set
	$B=2^k-2$, this immediately gives a $O^*\left((B-\epsilon)^n\right)$
	algorithm for \qcsp. However, we assumed that $q$ is sufficiently large
	that the existence of such an algorithm would contradict the SETH.
\end{proof}

%% file: figures/weak-edge.tex
\ifx\du\undefined
  \newlength{\du}
\fi
\setlength{\du}{15\unitlength}
\begin{tikzpicture}[even odd rule,scale=0.5]
\pgftransformxscale{1.000000}
\pgftransformyscale{-1.000000}
\definecolor{dialinecolor}{rgb}{0.000000, 0.000000, 0.000000}
\pgfsetstrokecolor{dialinecolor}
\pgfsetstrokeopacity{1.000000}
\definecolor{diafillcolor}{rgb}{1.000000, 1.000000, 1.000000}
\pgfsetfillcolor{diafillcolor}
\pgfsetfillopacity{1.000000}
\pgfsetlinewidth{0.100000\du}
\pgfsetdash{}{0pt}
\definecolor{diafillcolor}{rgb}{0.000000, 0.000000, 0.000000}
\pgfsetfillcolor{diafillcolor}
\pgfsetfillopacity{1.000000}
\pgfpathellipse{\pgfpoint{17.500000\du}{13.500000\du}}{\pgfpoint{0.500000\du}{0\du}}{\pgfpoint{0\du}{0.500000\du}}
\pgfusepath{fill}
\definecolor{dialinecolor}{rgb}{0.000000, 0.000000, 0.000000}
\pgfsetstrokecolor{dialinecolor}
\pgfsetstrokeopacity{1.000000}
\pgfpathellipse{\pgfpoint{17.500000\du}{13.500000\du}}{\pgfpoint{0.500000\du}{0\du}}{\pgfpoint{0\du}{0.500000\du}}
\pgfusepath{stroke}
\pgfsetlinewidth{0.100000\du}
\pgfsetdash{}{0pt}
\definecolor{diafillcolor}{rgb}{1.000000, 1.000000, 1.000000}
\pgfsetfillcolor{diafillcolor}
\pgfsetfillopacity{1.000000}
\pgfpathellipse{\pgfpoint{20.500000\du}{13.500000\du}}{\pgfpoint{0.500000\du}{0\du}}{\pgfpoint{0\du}{0.500000\du}}
\pgfusepath{fill}
\definecolor{dialinecolor}{rgb}{0.000000, 0.000000, 0.000000}
\pgfsetstrokecolor{dialinecolor}
\pgfsetstrokeopacity{1.000000}
\pgfpathellipse{\pgfpoint{20.500000\du}{13.500000\du}}{\pgfpoint{0.500000\du}{0\du}}{\pgfpoint{0\du}{0.500000\du}}
\pgfusepath{stroke}
\pgfsetlinewidth{0.100000\du}
\pgfsetdash{}{0pt}
\definecolor{diafillcolor}{rgb}{1.000000, 1.000000, 1.000000}
\pgfsetfillcolor{diafillcolor}
\pgfsetfillopacity{1.000000}
\pgfpathellipse{\pgfpoint{23.000000\du}{13.500000\du}}{\pgfpoint{0.500000\du}{0\du}}{\pgfpoint{0\du}{0.500000\du}}
\pgfusepath{fill}
\definecolor{dialinecolor}{rgb}{0.000000, 0.000000, 0.000000}
\pgfsetstrokecolor{dialinecolor}
\pgfsetstrokeopacity{1.000000}
\pgfpathellipse{\pgfpoint{23.000000\du}{13.500000\du}}{\pgfpoint{0.500000\du}{0\du}}{\pgfpoint{0\du}{0.500000\du}}
\pgfusepath{stroke}
\pgfsetlinewidth{0.100000\du}
\pgfsetdash{}{0pt}
\definecolor{diafillcolor}{rgb}{0.000000, 0.000000, 0.000000}
\pgfsetfillcolor{diafillcolor}
\pgfsetfillopacity{1.000000}
\pgfpathellipse{\pgfpoint{28.500000\du}{13.500000\du}}{\pgfpoint{0.500000\du}{0\du}}{\pgfpoint{0\du}{0.500000\du}}
\pgfusepath{fill}
\definecolor{dialinecolor}{rgb}{0.000000, 0.000000, 0.000000}
\pgfsetstrokecolor{dialinecolor}
\pgfsetstrokeopacity{1.000000}
\pgfpathellipse{\pgfpoint{28.500000\du}{13.500000\du}}{\pgfpoint{0.500000\du}{0\du}}{\pgfpoint{0\du}{0.500000\du}}
\pgfusepath{stroke}
\pgfsetlinewidth{0.100000\du}
\pgfsetdash{}{0pt}
\definecolor{diafillcolor}{rgb}{1.000000, 1.000000, 1.000000}
\pgfsetfillcolor{diafillcolor}
\pgfsetfillopacity{1.000000}
\pgfpathellipse{\pgfpoint{25.500000\du}{13.500000\du}}{\pgfpoint{0.500000\du}{0\du}}{\pgfpoint{0\du}{0.500000\du}}
\pgfusepath{fill}
\definecolor{dialinecolor}{rgb}{0.000000, 0.000000, 0.000000}
\pgfsetstrokecolor{dialinecolor}
\pgfsetstrokeopacity{1.000000}
\pgfpathellipse{\pgfpoint{25.500000\du}{13.500000\du}}{\pgfpoint{0.500000\du}{0\du}}{\pgfpoint{0\du}{0.500000\du}}
\pgfusepath{stroke}
\pgfsetlinewidth{0.100000\du}
\pgfsetdash{}{0pt}
\pgfsetbuttcap
{
\definecolor{diafillcolor}{rgb}{0.000000, 0.000000, 0.000000}
\pgfsetfillcolor{diafillcolor}
\pgfsetfillopacity{1.000000}
\definecolor{dialinecolor}{rgb}{0.000000, 0.000000, 0.000000}
\pgfsetstrokecolor{dialinecolor}
\pgfsetstrokeopacity{1.000000}
\draw (18.044922\du,13.500000\du)--(19.955078\du,13.500000\du);
}
\pgfsetlinewidth{0.100000\du}
\pgfsetdash{}{0pt}
\pgfsetbuttcap
{
\definecolor{diafillcolor}{rgb}{0.000000, 0.000000, 0.000000}
\pgfsetfillcolor{diafillcolor}
\pgfsetfillopacity{1.000000}
\definecolor{dialinecolor}{rgb}{0.000000, 0.000000, 0.000000}
\pgfsetstrokecolor{dialinecolor}
\pgfsetstrokeopacity{1.000000}
\draw (21.049927\du,13.500000\du)--(22.450073\du,13.500000\du);
}
\pgfsetlinewidth{0.100000\du}
\pgfsetdash{}{0pt}
\pgfsetbuttcap
{
\definecolor{diafillcolor}{rgb}{0.000000, 0.000000, 0.000000}
\pgfsetfillcolor{diafillcolor}
\pgfsetfillopacity{1.000000}
\definecolor{dialinecolor}{rgb}{0.000000, 0.000000, 0.000000}
\pgfsetstrokecolor{dialinecolor}
\pgfsetstrokeopacity{1.000000}
\draw (23.549927\du,13.500000\du)--(24.950073\du,13.500000\du);
}
\pgfsetlinewidth{0.100000\du}
\pgfsetdash{}{0pt}
\pgfsetbuttcap
{
\definecolor{diafillcolor}{rgb}{0.000000, 0.000000, 0.000000}
\pgfsetfillcolor{diafillcolor}
\pgfsetfillopacity{1.000000}
\definecolor{dialinecolor}{rgb}{0.000000, 0.000000, 0.000000}
\pgfsetstrokecolor{dialinecolor}
\pgfsetstrokeopacity{1.000000}
\draw (26.044922\du,13.500000\du)--(27.955078\du,13.500000\du);
}
\definecolor{dialinecolor}{rgb}{0.000000, 0.000000, 0.000000}
\pgfsetstrokecolor{dialinecolor}
\pgfsetstrokeopacity{1.000000}
\definecolor{diafillcolor}{rgb}{0.000000, 0.000000, 0.000000}
\pgfsetfillcolor{diafillcolor}
\pgfsetfillopacity{1.000000}
\node[anchor=base west,inner sep=0pt,outer sep=0pt,color=dialinecolor] at (17.000000\du,12.500000\du){};
\definecolor{dialinecolor}{rgb}{0.000000, 0.000000, 0.000000}
\pgfsetstrokecolor{dialinecolor}
\pgfsetstrokeopacity{1.000000}
\definecolor{diafillcolor}{rgb}{0.000000, 0.000000, 0.000000}
\pgfsetfillcolor{diafillcolor}
\pgfsetfillopacity{1.000000}
\node[anchor=base west,inner sep=0pt,outer sep=0pt,color=dialinecolor] at (19.000000\du,12.500000\du){\small$\{1,5\}$};
\definecolor{dialinecolor}{rgb}{0.000000, 0.000000, 0.000000}
\pgfsetstrokecolor{dialinecolor}
\pgfsetstrokeopacity{1.000000}
\definecolor{diafillcolor}{rgb}{0.000000, 0.000000, 0.000000}
\pgfsetfillcolor{diafillcolor}
\pgfsetfillopacity{1.000000}
\node[anchor=base west,inner sep=0pt,outer sep=0pt,color=dialinecolor] at (21.500000\du,15.000000\du){\small$\{2,5\}$};
\definecolor{dialinecolor}{rgb}{0.000000, 0.000000, 0.000000}
\pgfsetstrokecolor{dialinecolor}
\pgfsetstrokeopacity{1.000000}
\definecolor{diafillcolor}{rgb}{0.000000, 0.000000, 0.000000}
\pgfsetfillcolor{diafillcolor}
\pgfsetfillopacity{1.000000}
\node[anchor=base west,inner sep=0pt,outer sep=0pt,color=dialinecolor] at (24.000000\du,12.500000\du){\small$\{2,5\}$};
\definecolor{dialinecolor}{rgb}{0.000000, 0.000000, 0.000000}
\pgfsetstrokecolor{dialinecolor}
\pgfsetstrokeopacity{1.000000}
\definecolor{diafillcolor}{rgb}{0.000000, 0.000000, 0.000000}
\pgfsetfillcolor{diafillcolor}
\pgfsetfillopacity{1.000000}
\node[anchor=base west,inner sep=0pt,outer sep=0pt,color=dialinecolor] at (28.000000\du,12.500000\du){};
\pgfsetlinewidth{0.100000\du}
\pgfsetdash{}{0pt}
\definecolor{diafillcolor}{rgb}{0.000000, 0.000000, 0.000000}
\pgfsetfillcolor{diafillcolor}
\pgfsetfillopacity{1.000000}
\pgfpathellipse{\pgfpoint{17.500000\du}{20.000000\du}}{\pgfpoint{0.500000\du}{0\du}}{\pgfpoint{0\du}{0.500000\du}}
\pgfusepath{fill}
\definecolor{dialinecolor}{rgb}{0.000000, 0.000000, 0.000000}
\pgfsetstrokecolor{dialinecolor}
\pgfsetstrokeopacity{1.000000}
\pgfpathellipse{\pgfpoint{17.500000\du}{20.000000\du}}{\pgfpoint{0.500000\du}{0\du}}{\pgfpoint{0\du}{0.500000\du}}
\pgfusepath{stroke}
\pgfsetlinewidth{0.100000\du}
\pgfsetdash{}{0pt}
\definecolor{diafillcolor}{rgb}{0.000000, 0.000000, 0.000000}
\pgfsetfillcolor{diafillcolor}
\pgfsetfillopacity{1.000000}
\pgfpathellipse{\pgfpoint{28.500000\du}{20.000000\du}}{\pgfpoint{0.500000\du}{0\du}}{\pgfpoint{0\du}{0.500000\du}}
\pgfusepath{fill}
\definecolor{dialinecolor}{rgb}{0.000000, 0.000000, 0.000000}
\pgfsetstrokecolor{dialinecolor}
\pgfsetstrokeopacity{1.000000}
\pgfpathellipse{\pgfpoint{28.500000\du}{20.000000\du}}{\pgfpoint{0.500000\du}{0\du}}{\pgfpoint{0\du}{0.500000\du}}
\pgfusepath{stroke}
\definecolor{dialinecolor}{rgb}{0.000000, 0.000000, 0.000000}
\pgfsetstrokecolor{dialinecolor}
\pgfsetstrokeopacity{1.000000}
\definecolor{diafillcolor}{rgb}{0.000000, 0.000000, 0.000000}
\pgfsetfillcolor{diafillcolor}
\pgfsetfillopacity{1.000000}
\node[anchor=base west,inner sep=0pt,outer sep=0pt,color=dialinecolor] at (21.500000\du,19.000000\du){(1,5)};
\pgfsetlinewidth{0.100000\du}
\pgfsetdash{{0.500000\du}{0.500000\du}}{0\du}
\pgfsetbuttcap
{
\definecolor{diafillcolor}{rgb}{0.000000, 0.000000, 0.000000}
\pgfsetfillcolor{diafillcolor}
\pgfsetfillopacity{1.000000}
\definecolor{dialinecolor}{rgb}{0.000000, 0.000000, 0.000000}
\pgfsetstrokecolor{dialinecolor}
\pgfsetstrokeopacity{1.000000}
\draw (18.049194\du,20.000000\du)--(27.950806\du,20.000000\du);
}
\definecolor{dialinecolor}{rgb}{0.000000, 0.000000, 0.000000}
\pgfsetstrokecolor{dialinecolor}
\pgfsetstrokeopacity{1.000000}
\definecolor{diafillcolor}{rgb}{0.000000, 0.000000, 0.000000}
\pgfsetfillcolor{diafillcolor}
\pgfsetfillopacity{1.000000}
\node[anchor=base west,inner sep=0pt,outer sep=0pt,color=dialinecolor] at (22.500000\du,17.000000\du){\Large$\equiv$};
\end{tikzpicture}

%% file: figures/implication.tex
\ifx\du\undefined
  \newlength{\du}
\fi
\setlength{\du}{15\unitlength}
\begin{tikzpicture}[even odd rule,scale=0.5]
\pgftransformxscale{1.000000}
\pgftransformyscale{-1.000000}
\definecolor{dialinecolor}{rgb}{0.000000, 0.000000, 0.000000}
\pgfsetstrokecolor{dialinecolor}
\pgfsetstrokeopacity{1.000000}
\definecolor{diafillcolor}{rgb}{1.000000, 1.000000, 1.000000}
\pgfsetfillcolor{diafillcolor}
\pgfsetfillopacity{1.000000}
\pgfsetlinewidth{0.100000\du}
\pgfsetdash{}{0pt}
\definecolor{diafillcolor}{rgb}{0.000000, 0.000000, 0.000000}
\pgfsetfillcolor{diafillcolor}
\pgfsetfillopacity{1.000000}
\pgfpathellipse{\pgfpoint{17.500000\du}{20.000000\du}}{\pgfpoint{0.500000\du}{0\du}}{\pgfpoint{0\du}{0.500000\du}}
\pgfusepath{fill}
\definecolor{dialinecolor}{rgb}{0.000000, 0.000000, 0.000000}
\pgfsetstrokecolor{dialinecolor}
\pgfsetstrokeopacity{1.000000}
\pgfpathellipse{\pgfpoint{17.500000\du}{20.000000\du}}{\pgfpoint{0.500000\du}{0\du}}{\pgfpoint{0\du}{0.500000\du}}
\pgfusepath{stroke}
\pgfsetlinewidth{0.100000\du}
\pgfsetdash{}{0pt}
\definecolor{diafillcolor}{rgb}{0.000000, 0.000000, 0.000000}
\pgfsetfillcolor{diafillcolor}
\pgfsetfillopacity{1.000000}
\pgfpathellipse{\pgfpoint{28.500000\du}{20.000000\du}}{\pgfpoint{0.500000\du}{0\du}}{\pgfpoint{0\du}{0.500000\du}}
\pgfusepath{fill}
\definecolor{dialinecolor}{rgb}{0.000000, 0.000000, 0.000000}
\pgfsetstrokecolor{dialinecolor}
\pgfsetstrokeopacity{1.000000}
\pgfpathellipse{\pgfpoint{28.500000\du}{20.000000\du}}{\pgfpoint{0.500000\du}{0\du}}{\pgfpoint{0\du}{0.500000\du}}
\pgfusepath{stroke}
\definecolor{dialinecolor}{rgb}{0.000000, 0.000000, 0.000000}
\pgfsetstrokecolor{dialinecolor}
\pgfsetstrokeopacity{1.000000}
\definecolor{diafillcolor}{rgb}{0.000000, 0.000000, 0.000000}
\pgfsetfillcolor{diafillcolor}
\pgfsetfillopacity{1.000000}
	\node[anchor=base west,inner sep=0pt,outer sep=0pt,color=dialinecolor] at (20.500000\du,19.000000\du){$(1\to2)$};
\pgfsetlinewidth{0.100000\du}
\pgfsetdash{{0.100000\du}{0.100000\du}}{0\du}
\pgfsetbuttcap
{
\definecolor{diafillcolor}{rgb}{0.000000, 0.000000, 0.000000}
\pgfsetfillcolor{diafillcolor}
\pgfsetfillopacity{1.000000}
\pgfsetarrowsend{stealth}
\definecolor{dialinecolor}{rgb}{0.000000, 0.000000, 0.000000}
\pgfsetstrokecolor{dialinecolor}
\pgfsetstrokeopacity{1.000000}
\draw (18.049194\du,20.000000\du)--(27.950806\du,20.000000\du);
}
\definecolor{dialinecolor}{rgb}{0.000000, 0.000000, 0.000000}
\pgfsetstrokecolor{dialinecolor}
\pgfsetstrokeopacity{1.000000}
\definecolor{diafillcolor}{rgb}{0.000000, 0.000000, 0.000000}
\pgfsetfillcolor{diafillcolor}
\pgfsetfillopacity{1.000000}
\node[anchor=base west,inner sep=0pt,outer sep=0pt,color=dialinecolor] at (22.000000\du,17.000000\du){\Large$\equiv$};
\pgfsetlinewidth{0.100000\du}
\pgfsetdash{}{0pt}
\definecolor{diafillcolor}{rgb}{0.000000, 0.000000, 0.000000}
\pgfsetfillcolor{diafillcolor}
\pgfsetfillopacity{1.000000}
\pgfpathellipse{\pgfpoint{17.500000\du}{13.500000\du}}{\pgfpoint{0.500000\du}{0\du}}{\pgfpoint{0\du}{0.500000\du}}
\pgfusepath{fill}
\definecolor{dialinecolor}{rgb}{0.000000, 0.000000, 0.000000}
\pgfsetstrokecolor{dialinecolor}
\pgfsetstrokeopacity{1.000000}
\pgfpathellipse{\pgfpoint{17.500000\du}{13.500000\du}}{\pgfpoint{0.500000\du}{0\du}}{\pgfpoint{0\du}{0.500000\du}}
\pgfusepath{stroke}
\pgfsetlinewidth{0.100000\du}
\pgfsetdash{}{0pt}
\definecolor{diafillcolor}{rgb}{0.000000, 0.000000, 0.000000}
\pgfsetfillcolor{diafillcolor}
\pgfsetfillopacity{1.000000}
\pgfpathellipse{\pgfpoint{28.500000\du}{13.500000\du}}{\pgfpoint{0.500000\du}{0\du}}{\pgfpoint{0\du}{0.500000\du}}
\pgfusepath{fill}
\definecolor{dialinecolor}{rgb}{0.000000, 0.000000, 0.000000}
\pgfsetstrokecolor{dialinecolor}
\pgfsetstrokeopacity{1.000000}
\pgfpathellipse{\pgfpoint{28.500000\du}{13.500000\du}}{\pgfpoint{0.500000\du}{0\du}}{\pgfpoint{0\du}{0.500000\du}}
\pgfusepath{stroke}
\pgfsetlinewidth{0.100000\du}
\pgfsetdash{{0.500000\du}{0.500000\du}}{0\du}
\pgfsetbuttcap
{
\definecolor{diafillcolor}{rgb}{0.000000, 0.000000, 0.000000}
\pgfsetfillcolor{diafillcolor}
\pgfsetfillopacity{1.000000}
\definecolor{dialinecolor}{rgb}{0.000000, 0.000000, 0.000000}
\pgfsetstrokecolor{dialinecolor}
\pgfsetstrokeopacity{1.000000}
\pgfpathmoveto{\pgfpoint{27.950395\du}{13.499726\du}}
\pgfpatharc{304}{237}{8.920159\du and 8.920159\du}
\pgfusepath{stroke}
}
\pgfsetlinewidth{0.100000\du}
\pgfsetdash{{0.500000\du}{0.500000\du}}{0\du}
\pgfsetbuttcap
{
\definecolor{diafillcolor}{rgb}{0.000000, 0.000000, 0.000000}
\pgfsetfillcolor{diafillcolor}
\pgfsetfillopacity{1.000000}
\definecolor{dialinecolor}{rgb}{0.000000, 0.000000, 0.000000}
\pgfsetstrokecolor{dialinecolor}
\pgfsetstrokeopacity{1.000000}
\pgfpathmoveto{\pgfpoint{27.950139\du}{13.498723\du}}
\pgfpatharc{333}{208}{5.585079\du and 5.585079\du}
\pgfusepath{stroke}
}
\pgfsetlinewidth{0.100000\du}
\pgfsetdash{{0.500000\du}{0.500000\du}}{0\du}
\pgfsetbuttcap
{
\definecolor{diafillcolor}{rgb}{0.000000, 0.000000, 0.000000}
\pgfsetfillcolor{diafillcolor}
\pgfsetfillopacity{1.000000}
\definecolor{dialinecolor}{rgb}{0.000000, 0.000000, 0.000000}
\pgfsetstrokecolor{dialinecolor}
\pgfsetstrokeopacity{1.000000}
\pgfpathmoveto{\pgfpoint{18.049214\du}{13.500013\du}}
\pgfpatharc{124}{57}{8.920159\du and 8.920159\du}
\pgfusepath{stroke}
}
\definecolor{dialinecolor}{rgb}{0.000000, 0.000000, 0.000000}
\pgfsetstrokecolor{dialinecolor}
\pgfsetstrokeopacity{1.000000}
\definecolor{diafillcolor}{rgb}{0.000000, 0.000000, 0.000000}
\pgfsetfillcolor{diafillcolor}
\pgfsetfillopacity{1.000000}
\node[anchor=base west,inner sep=0pt,outer sep=0pt,color=dialinecolor] at (21.500000\du,9.500000\du){$(1,1)$};
\definecolor{dialinecolor}{rgb}{0.000000, 0.000000, 0.000000}
\pgfsetstrokecolor{dialinecolor}
\pgfsetstrokeopacity{1.000000}
\definecolor{diafillcolor}{rgb}{0.000000, 0.000000, 0.000000}
\pgfsetfillcolor{diafillcolor}
\pgfsetfillopacity{1.000000}
\node[anchor=base west,inner sep=0pt,outer sep=0pt,color=dialinecolor] at (21.500000\du,11.500000\du){$(1,3)$};
\definecolor{dialinecolor}{rgb}{0.000000, 0.000000, 0.000000}
\pgfsetstrokecolor{dialinecolor}
\pgfsetstrokeopacity{1.000000}
\definecolor{diafillcolor}{rgb}{0.000000, 0.000000, 0.000000}
\pgfsetfillcolor{diafillcolor}
\pgfsetfillopacity{1.000000}
\node[anchor=base west,inner sep=0pt,outer sep=0pt,color=dialinecolor] at (21.500000\du,14.500000\du){$(1,k)$};
\definecolor{dialinecolor}{rgb}{0.000000, 0.000000, 0.000000}
\pgfsetstrokecolor{dialinecolor}
\pgfsetstrokeopacity{1.000000}
\definecolor{diafillcolor}{rgb}{0.000000, 0.000000, 0.000000}
\pgfsetfillcolor{diafillcolor}
\pgfsetfillopacity{1.000000}
\node[anchor=base west,inner sep=0pt,outer sep=0pt,color=dialinecolor] at (22.500000\du,13.000000\du){\Large$\ldots$};
\end{tikzpicture}

%% file: figures/OR.tex
\ifx\du\undefined
  \newlength{\du}
\fi
\setlength{\du}{15\unitlength}
\begin{tikzpicture}[even odd rule,scale=0.5]
\pgftransformxscale{1.000000}
\pgftransformyscale{-1.000000}
\definecolor{dialinecolor}{rgb}{0.000000, 0.000000, 0.000000}
\pgfsetstrokecolor{dialinecolor}
\pgfsetstrokeopacity{1.000000}
\definecolor{diafillcolor}{rgb}{1.000000, 1.000000, 1.000000}
\pgfsetfillcolor{diafillcolor}
\pgfsetfillopacity{1.000000}
\pgfsetlinewidth{0.100000\du}
\pgfsetdash{}{0pt}
\definecolor{diafillcolor}{rgb}{0.000000, 0.000000, 0.000000}
\pgfsetfillcolor{diafillcolor}
\pgfsetfillopacity{1.000000}
\pgfpathellipse{\pgfpoint{17.500000\du}{13.500000\du}}{\pgfpoint{0.500000\du}{0\du}}{\pgfpoint{0\du}{0.500000\du}}
\pgfpathellipse{\pgfpoint{23.000000\du}{13.500000\du}}{\pgfpoint{0.500000\du}{0\du}}{\pgfpoint{0\du}{0.500000\du}}
\pgfusepath{fill}
\definecolor{dialinecolor}{rgb}{0.000000, 0.000000, 0.000000}
\pgfsetstrokecolor{dialinecolor}
\pgfsetstrokeopacity{1.000000}
\pgfpathellipse{\pgfpoint{17.500000\du}{13.500000\du}}{\pgfpoint{0.500000\du}{0\du}}{\pgfpoint{0\du}{0.500000\du}}
\pgfpathellipse{\pgfpoint{14.000000\du}{13.500000\du}}{\pgfpoint{0.500000\du}{0\du}}{\pgfpoint{0\du}{0.500000\du}}
\pgfpathellipse{\pgfpoint{32.000000\du}{13.500000\du}}{\pgfpoint{0.500000\du}{0\du}}{\pgfpoint{0\du}{0.500000\du}}
\pgfusepath{stroke}
\pgfsetlinewidth{0.100000\du}
\pgfsetdash{}{0pt}
\definecolor{diafillcolor}{rgb}{1.000000, 1.000000, 1.000000}
\pgfsetfillcolor{diafillcolor}
\pgfsetfillopacity{1.000000}
\pgfpathellipse{\pgfpoint{20.500000\du}{13.500000\du}}{\pgfpoint{0.500000\du}{0\du}}{\pgfpoint{0\du}{0.500000\du}}
\pgfusepath{fill}
\definecolor{dialinecolor}{rgb}{0.000000, 0.000000, 0.000000}
\pgfsetstrokecolor{dialinecolor}
\pgfsetstrokeopacity{1.000000}
\pgfpathellipse{\pgfpoint{20.500000\du}{13.500000\du}}{\pgfpoint{0.500000\du}{0\du}}{\pgfpoint{0\du}{0.500000\du}}
\pgfusepath{stroke}
\pgfsetlinewidth{0.100000\du}
\pgfsetdash{}{0pt}
\definecolor{diafillcolor}{rgb}{1.000000, 1.000000, 1.000000}
\pgfsetfillcolor{diafillcolor}
\pgfsetfillopacity{1.000000}
\pgfusepath{fill}
\definecolor{dialinecolor}{rgb}{0.000000, 0.000000, 0.000000}
\pgfsetstrokecolor{dialinecolor}
\pgfsetstrokeopacity{1.000000}
\pgfpathellipse{\pgfpoint{23.000000\du}{13.500000\du}}{\pgfpoint{0.500000\du}{0\du}}{\pgfpoint{0\du}{0.500000\du}}
\pgfusepath{stroke}
\pgfsetlinewidth{0.100000\du}
\pgfsetdash{}{0pt}
\definecolor{diafillcolor}{rgb}{0.000000, 0.000000, 0.000000}
\pgfsetfillcolor{diafillcolor}
\pgfsetfillopacity{1.000000}
\pgfpathellipse{\pgfpoint{28.500000\du}{13.500000\du}}{\pgfpoint{0.500000\du}{0\du}}{\pgfpoint{0\du}{0.500000\du}}
\pgfusepath{fill}
\definecolor{dialinecolor}{rgb}{0.000000, 0.000000, 0.000000}
\pgfsetstrokecolor{dialinecolor}
\pgfsetstrokeopacity{1.000000}
\pgfpathellipse{\pgfpoint{28.500000\du}{13.500000\du}}{\pgfpoint{0.500000\du}{0\du}}{\pgfpoint{0\du}{0.500000\du}}
\pgfusepath{stroke}
\pgfsetlinewidth{0.100000\du}
\pgfsetdash{}{0pt}
\definecolor{diafillcolor}{rgb}{1.000000, 1.000000, 1.000000}
\pgfsetfillcolor{diafillcolor}
\pgfsetfillopacity{1.000000}
\pgfpathellipse{\pgfpoint{25.500000\du}{13.500000\du}}{\pgfpoint{0.500000\du}{0\du}}{\pgfpoint{0\du}{0.500000\du}}
\pgfusepath{fill}
\definecolor{dialinecolor}{rgb}{0.000000, 0.000000, 0.000000}
\pgfsetstrokecolor{dialinecolor}
\pgfsetstrokeopacity{1.000000}
\pgfpathellipse{\pgfpoint{25.500000\du}{13.500000\du}}{\pgfpoint{0.500000\du}{0\du}}{\pgfpoint{0\du}{0.500000\du}}
\pgfusepath{stroke}
\pgfsetlinewidth{0.100000\du}
\pgfsetdash{}{0pt}
\pgfsetbuttcap
{
\definecolor{diafillcolor}{rgb}{0.000000, 0.000000, 0.000000}
\pgfsetfillcolor{diafillcolor}
\pgfsetfillopacity{1.000000}
\definecolor{dialinecolor}{rgb}{0.000000, 0.000000, 0.000000}
\pgfsetstrokecolor{dialinecolor}
\pgfsetstrokeopacity{1.000000}
\draw (18.044922\du,13.500000\du)--(19.955078\du,13.500000\du);
}
{
\definecolor{diafillcolor}{rgb}{0.000000, 0.000000, 0.000000}
\pgfsetfillcolor{diafillcolor}
\pgfsetfillopacity{1.000000}
\definecolor{dialinecolor}{rgb}{0.000000, 0.000000, 0.000000}
\pgfsetstrokecolor{dialinecolor}
\pgfsetstrokeopacity{1.000000}
\draw (14.544922\du,13.500000\du)--(17.955078\du,13.500000\du);
}
{
\definecolor{diafillcolor}{rgb}{0.000000, 0.000000, 0.000000}
\pgfsetfillcolor{diafillcolor}
\pgfsetfillopacity{1.000000}
\definecolor{dialinecolor}{rgb}{0.000000, 0.000000, 0.000000}
\pgfsetstrokecolor{dialinecolor}
\pgfsetstrokeopacity{1.000000}
\draw (28.544922\du,13.500000\du)--(31.455078\du,13.500000\du);
}
\pgfsetlinewidth{0.100000\du}
\pgfsetdash{}{0pt}
\pgfsetbuttcap
{
\definecolor{diafillcolor}{rgb}{0.000000, 0.000000, 0.000000}
\pgfsetfillcolor{diafillcolor}
\pgfsetfillopacity{1.000000}
\definecolor{dialinecolor}{rgb}{0.000000, 0.000000, 0.000000}
\pgfsetstrokecolor{dialinecolor}
\pgfsetstrokeopacity{1.000000}
\draw (21.049927\du,13.500000\du)--(22.450073\du,13.500000\du);
}
\pgfsetlinewidth{0.100000\du}
\pgfsetdash{}{0pt}
\pgfsetbuttcap
{
\definecolor{diafillcolor}{rgb}{0.000000, 0.000000, 0.000000}
\pgfsetfillcolor{diafillcolor}
\pgfsetfillopacity{1.000000}
\definecolor{dialinecolor}{rgb}{0.000000, 0.000000, 0.000000}
\pgfsetstrokecolor{dialinecolor}
\pgfsetstrokeopacity{1.000000}
\draw (23.549927\du,13.500000\du)--(24.950073\du,13.500000\du);
}
\pgfsetlinewidth{0.100000\du}
\pgfsetdash{}{0pt}
\pgfsetbuttcap
{
\definecolor{diafillcolor}{rgb}{0.000000, 0.000000, 0.000000}
\pgfsetfillcolor{diafillcolor}
\pgfsetfillopacity{1.000000}
\definecolor{dialinecolor}{rgb}{0.000000, 0.000000, 0.000000}
\pgfsetstrokecolor{dialinecolor}
\pgfsetstrokeopacity{1.000000}
\draw (26.044922\du,13.500000\du)--(27.955078\du,13.500000\du);
}
\definecolor{dialinecolor}{rgb}{0.000000, 0.000000, 0.000000}
\pgfsetstrokecolor{dialinecolor}
\pgfsetstrokeopacity{1.000000}
\definecolor{diafillcolor}{rgb}{0.000000, 0.000000, 0.000000}
\pgfsetfillcolor{diafillcolor}
\pgfsetfillopacity{1.000000}
\node[anchor=base west,inner sep=0pt,outer sep=0pt,color=dialinecolor] at (17.000000\du,12.500000\du){};
\definecolor{dialinecolor}{rgb}{0.000000, 0.000000, 0.000000}
\pgfsetstrokecolor{dialinecolor}
\pgfsetstrokeopacity{1.000000}
\definecolor{diafillcolor}{rgb}{0.000000, 0.000000, 0.000000}
\pgfsetfillcolor{diafillcolor}
\pgfsetfillopacity{1.000000}
\node[anchor=base west,inner sep=0pt,outer sep=0pt,color=dialinecolor] at (19.000000\du,12.500000\du){\small$\{2,3\}$};
\definecolor{dialinecolor}{rgb}{0.000000, 0.000000, 0.000000}
\pgfsetstrokecolor{dialinecolor}
\pgfsetstrokeopacity{1.000000}
\definecolor{diafillcolor}{rgb}{0.000000, 0.000000, 0.000000}
\pgfsetfillcolor{diafillcolor}
\pgfsetfillopacity{1.000000}
\node[anchor=base west,inner sep=0pt,outer sep=0pt,color=dialinecolor] at (21.000000\du,15.000000\du){\small$\{1,2,3\}$};
\node[anchor=base west,inner sep=0pt,outer sep=0pt,color=dialinecolor] at (15.500000\du,15.000000\du){\small$\{1,2,3\}$};
\node[anchor=base west,inner sep=0pt,outer sep=0pt,color=dialinecolor] at (26.500000\du,15.000000\du){\small$\{1,2,3\}$};
\definecolor{dialinecolor}{rgb}{0.000000, 0.000000, 0.000000}
\pgfsetstrokecolor{dialinecolor}
\pgfsetstrokeopacity{1.000000}
\definecolor{diafillcolor}{rgb}{0.000000, 0.000000, 0.000000}
\pgfsetfillcolor{diafillcolor}
\pgfsetfillopacity{1.000000}
\node[anchor=base west,inner sep=0pt,outer sep=0pt,color=dialinecolor] at (24.000000\du,12.500000\du){\small$\{2,3\}$};
\node[anchor=base west,inner sep=0pt,outer sep=0pt,color=dialinecolor] at (31.000000\du,12.500000\du){\small$\{3\}$};
\node[anchor=base west,inner sep=0pt,outer sep=0pt,color=dialinecolor] at (13.000000\du,12.500000\du){\small$\{2\}$};
\end{tikzpicture}

%% file: figures/red1.tex
\ifx\du\undefined
  \newlength{\du}
\fi
\setlength{\du}{15\unitlength}
\begin{tikzpicture}[even odd rule,scale=0.15]
\pgftransformxscale{1.000000}
\pgftransformyscale{-1.000000}
\definecolor{dialinecolor}{rgb}{0.000000, 0.000000, 0.000000}
\pgfsetstrokecolor{dialinecolor}
\pgfsetstrokeopacity{1.000000}
\definecolor{diafillcolor}{rgb}{1.000000, 1.000000, 1.000000}
\pgfsetfillcolor{diafillcolor}
\pgfsetfillopacity{1.000000}
\pgfsetlinewidth{0.010000\du}
\pgfsetdash{}{0pt}
\pgfsetmiterjoin
\pgfsetbuttcap
{\pgfsetcornersarced{\pgfpoint{0.000000\du}{0.000000\du}}\definecolor{diafillcolor}{rgb}{1.000000, 1.000000, 1.000000}
\pgfsetfillcolor{diafillcolor}
\pgfsetfillopacity{1.000000}
\fill (14.500000\du,11.000000\du)--(14.500000\du,16.000000\du)--(19.500000\du,16.000000\du)--(19.500000\du,11.000000\du)--cycle;
}{\pgfsetcornersarced{\pgfpoint{0.000000\du}{0.000000\du}}\definecolor{dialinecolor}{rgb}{0.000000, 0.000000, 0.000000}
\pgfsetstrokecolor{dialinecolor}
\pgfsetstrokeopacity{1.000000}
\draw (14.500000\du,11.000000\du)--(14.500000\du,16.000000\du)--(19.500000\du,16.000000\du)--(19.500000\du,11.000000\du)--cycle;
}\pgfsetlinewidth{0.010000\du}
\pgfsetdash{}{0pt}
\pgfsetbuttcap
{
\definecolor{diafillcolor}{rgb}{0.000000, 0.000000, 0.000000}
\pgfsetfillcolor{diafillcolor}
\pgfsetfillopacity{1.000000}
\definecolor{dialinecolor}{rgb}{0.000000, 0.000000, 0.000000}
\pgfsetstrokecolor{dialinecolor}
\pgfsetstrokeopacity{1.000000}
\draw (14.500000\du,13.500000\du)--(19.500000\du,13.500000\du);
}
\pgfsetlinewidth{0.010000\du}
\pgfsetdash{}{0pt}
\pgfsetbuttcap
{
\definecolor{diafillcolor}{rgb}{0.000000, 0.000000, 0.000000}
\pgfsetfillcolor{diafillcolor}
\pgfsetfillopacity{1.000000}
\definecolor{dialinecolor}{rgb}{0.000000, 0.000000, 0.000000}
\pgfsetstrokecolor{dialinecolor}
\pgfsetstrokeopacity{1.000000}
\draw (17.000000\du,11.000000\du)--(17.000000\du,16.000000\du);
}
\pgfsetlinewidth{0.010000\du}
\pgfsetdash{}{0pt}
\pgfsetmiterjoin
\pgfsetbuttcap
{\pgfsetcornersarced{\pgfpoint{0.000000\du}{0.000000\du}}\definecolor{diafillcolor}{rgb}{1.000000, 1.000000, 1.000000}
\pgfsetfillcolor{diafillcolor}
\pgfsetfillopacity{1.000000}
\fill (14.500000\du,18.500000\du)--(14.500000\du,23.500000\du)--(19.500000\du,23.500000\du)--(19.500000\du,18.500000\du)--cycle;
}{\pgfsetcornersarced{\pgfpoint{0.000000\du}{0.000000\du}}\definecolor{dialinecolor}{rgb}{0.000000, 0.000000, 0.000000}
\pgfsetstrokecolor{dialinecolor}
\pgfsetstrokeopacity{1.000000}
\draw (14.500000\du,18.500000\du)--(14.500000\du,23.500000\du)--(19.500000\du,23.500000\du)--(19.500000\du,18.500000\du)--cycle;
}\pgfsetlinewidth{0.010000\du}
\pgfsetdash{}{0pt}
\pgfsetbuttcap
{
\definecolor{diafillcolor}{rgb}{0.000000, 0.000000, 0.000000}
\pgfsetfillcolor{diafillcolor}
\pgfsetfillopacity{1.000000}
\definecolor{dialinecolor}{rgb}{0.000000, 0.000000, 0.000000}
\pgfsetstrokecolor{dialinecolor}
\pgfsetstrokeopacity{1.000000}
\draw (14.500000\du,21.000000\du)--(19.500000\du,21.000000\du);
}
\pgfsetlinewidth{0.010000\du}
\pgfsetdash{}{0pt}
\pgfsetbuttcap
{
\definecolor{diafillcolor}{rgb}{0.000000, 0.000000, 0.000000}
\pgfsetfillcolor{diafillcolor}
\pgfsetfillopacity{1.000000}
\definecolor{dialinecolor}{rgb}{0.000000, 0.000000, 0.000000}
\pgfsetstrokecolor{dialinecolor}
\pgfsetstrokeopacity{1.000000}
\draw (17.000000\du,18.500000\du)--(17.000000\du,23.500000\du);
}
\pgfsetlinewidth{0.010000\du}
\pgfsetdash{}{0pt}
\definecolor{diafillcolor}{rgb}{1.000000, 1.000000, 1.000000}
\pgfsetfillcolor{diafillcolor}
\pgfsetfillopacity{1.000000}
\pgfpathellipse{\pgfpoint{17.000000\du}{5.000000\du}}{\pgfpoint{5.000000\du}{0\du}}{\pgfpoint{0\du}{3.000000\du}}
\pgfusepath{fill}
\definecolor{dialinecolor}{rgb}{0.000000, 0.000000, 0.000000}
\pgfsetstrokecolor{dialinecolor}
\pgfsetstrokeopacity{1.000000}
\pgfpathellipse{\pgfpoint{17.000000\du}{5.000000\du}}{\pgfpoint{5.000000\du}{0\du}}{\pgfpoint{0\du}{3.000000\du}}
\pgfusepath{stroke}
\pgfsetlinewidth{0.010000\du}
\pgfsetdash{}{0pt}
\definecolor{diafillcolor}{rgb}{0.000000, 0.000000, 0.000000}
\pgfsetfillcolor{diafillcolor}
\pgfsetfillopacity{1.000000}
\pgfpathellipse{\pgfpoint{14.000000\du}{5.000000\du}}{\pgfpoint{0.500000\du}{0\du}}{\pgfpoint{0\du}{0.500000\du}}
\pgfusepath{fill}
\definecolor{dialinecolor}{rgb}{0.000000, 0.000000, 0.000000}
\pgfsetstrokecolor{dialinecolor}
\pgfsetstrokeopacity{1.000000}
\pgfpathellipse{\pgfpoint{14.000000\du}{5.000000\du}}{\pgfpoint{0.500000\du}{0\du}}{\pgfpoint{0\du}{0.500000\du}}
\pgfusepath{stroke}
\pgfsetlinewidth{0.010000\du}
\pgfsetdash{}{0pt}
\definecolor{diafillcolor}{rgb}{0.000000, 0.000000, 0.000000}
\pgfsetfillcolor{diafillcolor}
\pgfsetfillopacity{1.000000}
\pgfpathellipse{\pgfpoint{16.000000\du}{5.000000\du}}{\pgfpoint{0.500000\du}{0\du}}{\pgfpoint{0\du}{0.500000\du}}
\pgfusepath{fill}
\definecolor{dialinecolor}{rgb}{0.000000, 0.000000, 0.000000}
\pgfsetstrokecolor{dialinecolor}
\pgfsetstrokeopacity{1.000000}
\pgfpathellipse{\pgfpoint{16.000000\du}{5.000000\du}}{\pgfpoint{0.500000\du}{0\du}}{\pgfpoint{0\du}{0.500000\du}}
\pgfusepath{stroke}
\pgfsetlinewidth{0.010000\du}
\pgfsetdash{}{0pt}
\definecolor{diafillcolor}{rgb}{0.000000, 0.000000, 0.000000}
\pgfsetfillcolor{diafillcolor}
\pgfsetfillopacity{1.000000}
\pgfpathellipse{\pgfpoint{18.000000\du}{5.000000\du}}{\pgfpoint{0.500000\du}{0\du}}{\pgfpoint{0\du}{0.500000\du}}
\pgfusepath{fill}
\definecolor{dialinecolor}{rgb}{0.000000, 0.000000, 0.000000}
\pgfsetstrokecolor{dialinecolor}
\pgfsetstrokeopacity{1.000000}
\pgfpathellipse{\pgfpoint{18.000000\du}{5.000000\du}}{\pgfpoint{0.500000\du}{0\du}}{\pgfpoint{0\du}{0.500000\du}}
\pgfusepath{stroke}
\pgfsetlinewidth{0.010000\du}
\pgfsetdash{}{0pt}
\definecolor{diafillcolor}{rgb}{0.000000, 0.000000, 0.000000}
\pgfsetfillcolor{diafillcolor}
\pgfsetfillopacity{1.000000}
\pgfpathellipse{\pgfpoint{20.000000\du}{5.000000\du}}{\pgfpoint{0.500000\du}{0\du}}{\pgfpoint{0\du}{0.500000\du}}
\pgfusepath{fill}
\definecolor{dialinecolor}{rgb}{0.000000, 0.000000, 0.000000}
\pgfsetstrokecolor{dialinecolor}
\pgfsetstrokeopacity{1.000000}
\pgfpathellipse{\pgfpoint{20.000000\du}{5.000000\du}}{\pgfpoint{0.500000\du}{0\du}}{\pgfpoint{0\du}{0.500000\du}}
\pgfusepath{stroke}
\pgfsetlinewidth{0.010000\du}
\pgfsetdash{}{0pt}
\pgfsetmiterjoin
\pgfsetbuttcap
{\pgfsetcornersarced{\pgfpoint{0.000000\du}{0.000000\du}}\definecolor{diafillcolor}{rgb}{1.000000, 1.000000, 1.000000}
\pgfsetfillcolor{diafillcolor}
\pgfsetfillopacity{1.000000}
\fill (14.500000\du,28.500000\du)--(14.500000\du,33.500000\du)--(19.500000\du,33.500000\du)--(19.500000\du,28.500000\du)--cycle;
}{\pgfsetcornersarced{\pgfpoint{0.000000\du}{0.000000\du}}\definecolor{dialinecolor}{rgb}{0.000000, 0.000000, 0.000000}
\pgfsetstrokecolor{dialinecolor}
\pgfsetstrokeopacity{1.000000}
\draw (14.500000\du,28.500000\du)--(14.500000\du,33.500000\du)--(19.500000\du,33.500000\du)--(19.500000\du,28.500000\du)--cycle;
}\pgfsetlinewidth{0.010000\du}
\pgfsetdash{}{0pt}
\pgfsetbuttcap
{
\definecolor{diafillcolor}{rgb}{0.000000, 0.000000, 0.000000}
\pgfsetfillcolor{diafillcolor}
\pgfsetfillopacity{1.000000}
\definecolor{dialinecolor}{rgb}{0.000000, 0.000000, 0.000000}
\pgfsetstrokecolor{dialinecolor}
\pgfsetstrokeopacity{1.000000}
\draw (14.500000\du,31.000000\du)--(19.500000\du,31.000000\du);
}
\pgfsetlinewidth{0.010000\du}
\pgfsetdash{}{0pt}
\pgfsetbuttcap
{
\definecolor{diafillcolor}{rgb}{0.000000, 0.000000, 0.000000}
\pgfsetfillcolor{diafillcolor}
\pgfsetfillopacity{1.000000}
\definecolor{dialinecolor}{rgb}{0.000000, 0.000000, 0.000000}
\pgfsetstrokecolor{dialinecolor}
\pgfsetstrokeopacity{1.000000}
\draw (17.000000\du,28.500000\du)--(17.000000\du,33.500000\du);
}
\pgfsetlinewidth{0.010000\du}
\pgfsetdash{}{0pt}
\pgfsetmiterjoin
\pgfsetbuttcap
{\pgfsetcornersarced{\pgfpoint{0.000000\du}{0.000000\du}}\definecolor{diafillcolor}{rgb}{1.000000, 1.000000, 1.000000}
\pgfsetfillcolor{diafillcolor}
\pgfsetfillopacity{1.000000}
\fill (14.500000\du,36.000000\du)--(14.500000\du,41.000000\du)--(19.500000\du,41.000000\du)--(19.500000\du,36.000000\du)--cycle;
}{\pgfsetcornersarced{\pgfpoint{0.000000\du}{0.000000\du}}\definecolor{dialinecolor}{rgb}{0.000000, 0.000000, 0.000000}
\pgfsetstrokecolor{dialinecolor}
\pgfsetstrokeopacity{1.000000}
\draw (14.500000\du,36.000000\du)--(14.500000\du,41.000000\du)--(19.500000\du,41.000000\du)--(19.500000\du,36.000000\du)--cycle;
}\pgfsetlinewidth{0.010000\du}
\pgfsetdash{}{0pt}
\pgfsetbuttcap
{
\definecolor{diafillcolor}{rgb}{0.000000, 0.000000, 0.000000}
\pgfsetfillcolor{diafillcolor}
\pgfsetfillopacity{1.000000}
\definecolor{dialinecolor}{rgb}{0.000000, 0.000000, 0.000000}
\pgfsetstrokecolor{dialinecolor}
\pgfsetstrokeopacity{1.000000}
\draw (14.500000\du,38.500000\du)--(19.500000\du,38.500000\du);
}
\pgfsetlinewidth{0.010000\du}
\pgfsetdash{}{0pt}
\pgfsetbuttcap
{
\definecolor{diafillcolor}{rgb}{0.000000, 0.000000, 0.000000}
\pgfsetfillcolor{diafillcolor}
\pgfsetfillopacity{1.000000}
\definecolor{dialinecolor}{rgb}{0.000000, 0.000000, 0.000000}
\pgfsetstrokecolor{dialinecolor}
\pgfsetstrokeopacity{1.000000}
\draw (17.000000\du,36.000000\du)--(17.000000\du,41.000000\du);
}
\pgfsetlinewidth{0.010000\du}
\pgfsetdash{}{0pt}
\pgfsetmiterjoin
\pgfsetbuttcap
{\pgfsetcornersarced{\pgfpoint{0.000000\du}{0.000000\du}}\definecolor{diafillcolor}{rgb}{1.000000, 1.000000, 1.000000}
\pgfsetfillcolor{diafillcolor}
\pgfsetfillopacity{1.000000}
\fill (32.000000\du,18.500000\du)--(32.000000\du,23.500000\du)--(37.000000\du,23.500000\du)--(37.000000\du,18.500000\du)--cycle;
}{\pgfsetcornersarced{\pgfpoint{0.000000\du}{0.000000\du}}\definecolor{dialinecolor}{rgb}{0.000000, 0.000000, 0.000000}
\pgfsetstrokecolor{dialinecolor}
\pgfsetstrokeopacity{1.000000}
\draw (32.000000\du,18.500000\du)--(32.000000\du,23.500000\du)--(37.000000\du,23.500000\du)--(37.000000\du,18.500000\du)--cycle;
}\pgfsetlinewidth{0.010000\du}
\pgfsetdash{}{0pt}
\definecolor{diafillcolor}{rgb}{1.000000, 1.000000, 1.000000}
\pgfsetfillcolor{diafillcolor}
\pgfsetfillopacity{1.000000}
\pgfpathellipse{\pgfpoint{34.500000\du}{5.000000\du}}{\pgfpoint{5.000000\du}{0\du}}{\pgfpoint{0\du}{3.000000\du}}
\pgfusepath{fill}
\definecolor{dialinecolor}{rgb}{0.000000, 0.000000, 0.000000}
\pgfsetstrokecolor{dialinecolor}
\pgfsetstrokeopacity{1.000000}
\pgfpathellipse{\pgfpoint{34.500000\du}{5.000000\du}}{\pgfpoint{5.000000\du}{0\du}}{\pgfpoint{0\du}{3.000000\du}}
\pgfusepath{stroke}
\pgfsetlinewidth{0.010000\du}
\pgfsetdash{}{0pt}
\pgfsetmiterjoin
\pgfsetbuttcap
{\pgfsetcornersarced{\pgfpoint{0.000000\du}{0.000000\du}}\definecolor{diafillcolor}{rgb}{1.000000, 1.000000, 1.000000}
\pgfsetfillcolor{diafillcolor}
\pgfsetfillopacity{1.000000}
\fill (32.000000\du,36.000000\du)--(32.000000\du,41.000000\du)--(37.000000\du,41.000000\du)--(37.000000\du,36.000000\du)--cycle;
}{\pgfsetcornersarced{\pgfpoint{0.000000\du}{0.000000\du}}\definecolor{dialinecolor}{rgb}{0.000000, 0.000000, 0.000000}
\pgfsetstrokecolor{dialinecolor}
\pgfsetstrokeopacity{1.000000}
\draw (32.000000\du,36.000000\du)--(32.000000\du,41.000000\du)--(37.000000\du,41.000000\du)--(37.000000\du,36.000000\du)--cycle;
}\pgfsetlinewidth{0.010000\du}
\pgfsetdash{}{0pt}
\definecolor{diafillcolor}{rgb}{1.000000, 1.000000, 1.000000}
\pgfsetfillcolor{diafillcolor}
\pgfsetfillopacity{1.000000}
\pgfpathellipse{\pgfpoint{50.500000\du}{5.000000\du}}{\pgfpoint{5.000000\du}{0\du}}{\pgfpoint{0\du}{3.000000\du}}
\pgfusepath{fill}
\definecolor{dialinecolor}{rgb}{0.000000, 0.000000, 0.000000}
\pgfsetstrokecolor{dialinecolor}
\pgfsetstrokeopacity{1.000000}
\pgfpathellipse{\pgfpoint{50.500000\du}{5.000000\du}}{\pgfpoint{5.000000\du}{0\du}}{\pgfpoint{0\du}{3.000000\du}}
\pgfusepath{stroke}
\pgfsetlinewidth{0.010000\du}
\pgfsetdash{}{0pt}
\pgfsetmiterjoin
\pgfsetbuttcap
{\pgfsetcornersarced{\pgfpoint{0.000000\du}{0.000000\du}}\definecolor{diafillcolor}{rgb}{1.000000, 1.000000, 1.000000}
\pgfsetfillcolor{diafillcolor}
\pgfsetfillopacity{1.000000}
\fill (48.000000\du,36.000000\du)--(48.000000\du,41.000000\du)--(53.000000\du,41.000000\du)--(53.000000\du,36.000000\du)--cycle;
}{\pgfsetcornersarced{\pgfpoint{0.000000\du}{0.000000\du}}\definecolor{dialinecolor}{rgb}{0.000000, 0.000000, 0.000000}
\pgfsetstrokecolor{dialinecolor}
\pgfsetstrokeopacity{1.000000}
\draw (48.000000\du,36.000000\du)--(48.000000\du,41.000000\du)--(53.000000\du,41.000000\du)--(53.000000\du,36.000000\du)--cycle;
}\pgfsetlinewidth{0.010000\du}
\pgfsetdash{}{0pt}
\pgfsetmiterjoin
\pgfsetbuttcap
{\pgfsetcornersarced{\pgfpoint{0.000000\du}{0.000000\du}}\definecolor{diafillcolor}{rgb}{1.000000, 1.000000, 1.000000}
\pgfsetfillcolor{diafillcolor}
\pgfsetfillopacity{1.000000}
\fill (64.000000\du,11.000000\du)--(64.000000\du,16.000000\du)--(69.000000\du,16.000000\du)--(69.000000\du,11.000000\du)--cycle;
}{\pgfsetcornersarced{\pgfpoint{0.000000\du}{0.000000\du}}\definecolor{dialinecolor}{rgb}{0.000000, 0.000000, 0.000000}
\pgfsetstrokecolor{dialinecolor}
\pgfsetstrokeopacity{1.000000}
\draw (64.000000\du,11.000000\du)--(64.000000\du,16.000000\du)--(69.000000\du,16.000000\du)--(69.000000\du,11.000000\du)--cycle;
}\pgfsetlinewidth{0.010000\du}
\pgfsetdash{}{0pt}
\definecolor{diafillcolor}{rgb}{1.000000, 1.000000, 1.000000}
\pgfsetfillcolor{diafillcolor}
\pgfsetfillopacity{1.000000}
\pgfpathellipse{\pgfpoint{66.500000\du}{5.000000\du}}{\pgfpoint{5.000000\du}{0\du}}{\pgfpoint{0\du}{3.000000\du}}
\pgfusepath{fill}
\definecolor{dialinecolor}{rgb}{0.000000, 0.000000, 0.000000}
\pgfsetstrokecolor{dialinecolor}
\pgfsetstrokeopacity{1.000000}
\pgfpathellipse{\pgfpoint{66.500000\du}{5.000000\du}}{\pgfpoint{5.000000\du}{0\du}}{\pgfpoint{0\du}{3.000000\du}}
\pgfusepath{stroke}
\pgfsetlinewidth{0.010000\du}
\pgfsetdash{}{0pt}
\pgfsetmiterjoin
\pgfsetbuttcap
{\pgfsetcornersarced{\pgfpoint{0.000000\du}{0.000000\du}}\definecolor{diafillcolor}{rgb}{1.000000, 1.000000, 1.000000}
\pgfsetfillcolor{diafillcolor}
\pgfsetfillopacity{1.000000}
\fill (80.000000\du,11.000000\du)--(80.000000\du,16.000000\du)--(85.000000\du,16.000000\du)--(85.000000\du,11.000000\du)--cycle;
}{\pgfsetcornersarced{\pgfpoint{0.000000\du}{0.000000\du}}\definecolor{dialinecolor}{rgb}{0.000000, 0.000000, 0.000000}
\pgfsetstrokecolor{dialinecolor}
\pgfsetstrokeopacity{1.000000}
\draw (80.000000\du,11.000000\du)--(80.000000\du,16.000000\du)--(85.000000\du,16.000000\du)--(85.000000\du,11.000000\du)--cycle;
}\pgfsetlinewidth{0.010000\du}
\pgfsetdash{}{0pt}
\pgfsetmiterjoin
\pgfsetbuttcap
{\pgfsetcornersarced{\pgfpoint{0.000000\du}{0.000000\du}}\definecolor{diafillcolor}{rgb}{1.000000, 1.000000, 1.000000}
\pgfsetfillcolor{diafillcolor}
\pgfsetfillopacity{1.000000}
\fill (80.000000\du,18.500000\du)--(80.000000\du,23.500000\du)--(85.000000\du,23.500000\du)--(85.000000\du,18.500000\du)--cycle;
}{\pgfsetcornersarced{\pgfpoint{0.000000\du}{0.000000\du}}\definecolor{dialinecolor}{rgb}{0.000000, 0.000000, 0.000000}
\pgfsetstrokecolor{dialinecolor}
\pgfsetstrokeopacity{1.000000}
\draw (80.000000\du,18.500000\du)--(80.000000\du,23.500000\du)--(85.000000\du,23.500000\du)--(85.000000\du,18.500000\du)--cycle;
}\pgfsetlinewidth{0.010000\du}
\pgfsetdash{}{0pt}
\definecolor{diafillcolor}{rgb}{1.000000, 1.000000, 1.000000}
\pgfsetfillcolor{diafillcolor}
\pgfsetfillopacity{1.000000}
\pgfpathellipse{\pgfpoint{82.500000\du}{5.000000\du}}{\pgfpoint{5.000000\du}{0\du}}{\pgfpoint{0\du}{3.000000\du}}
\pgfusepath{fill}
\definecolor{dialinecolor}{rgb}{0.000000, 0.000000, 0.000000}
\pgfsetstrokecolor{dialinecolor}
\pgfsetstrokeopacity{1.000000}
\pgfpathellipse{\pgfpoint{82.500000\du}{5.000000\du}}{\pgfpoint{5.000000\du}{0\du}}{\pgfpoint{0\du}{3.000000\du}}
\pgfusepath{stroke}
\pgfsetlinewidth{0.010000\du}
\pgfsetdash{}{0pt}
\pgfsetmiterjoin
\pgfsetbuttcap
{\pgfsetcornersarced{\pgfpoint{0.000000\du}{0.000000\du}}\definecolor{diafillcolor}{rgb}{1.000000, 1.000000, 1.000000}
\pgfsetfillcolor{diafillcolor}
\pgfsetfillopacity{1.000000}
\fill (80.000000\du,28.500000\du)--(80.000000\du,33.500000\du)--(85.000000\du,33.500000\du)--(85.000000\du,28.500000\du)--cycle;
}{\pgfsetcornersarced{\pgfpoint{0.000000\du}{0.000000\du}}\definecolor{dialinecolor}{rgb}{0.000000, 0.000000, 0.000000}
\pgfsetstrokecolor{dialinecolor}
\pgfsetstrokeopacity{1.000000}
\draw (80.000000\du,28.500000\du)--(80.000000\du,33.500000\du)--(85.000000\du,33.500000\du)--(85.000000\du,28.500000\du)--cycle;
}\pgfsetlinewidth{0.010000\du}
\pgfsetdash{}{0pt}
\pgfsetmiterjoin
\pgfsetbuttcap
{\pgfsetcornersarced{\pgfpoint{0.000000\du}{0.000000\du}}\definecolor{diafillcolor}{rgb}{1.000000, 1.000000, 1.000000}
\pgfsetfillcolor{diafillcolor}
\pgfsetfillopacity{1.000000}
\fill (80.000000\du,36.000000\du)--(80.000000\du,41.000000\du)--(85.000000\du,41.000000\du)--(85.000000\du,36.000000\du)--cycle;
}{\pgfsetcornersarced{\pgfpoint{0.000000\du}{0.000000\du}}\definecolor{dialinecolor}{rgb}{0.000000, 0.000000, 0.000000}
\pgfsetstrokecolor{dialinecolor}
\pgfsetstrokeopacity{1.000000}
\draw (80.000000\du,36.000000\du)--(80.000000\du,41.000000\du)--(85.000000\du,41.000000\du)--(85.000000\du,36.000000\du)--cycle;
}\pgfsetlinewidth{0.100000\du}
\pgfsetdash{}{0pt}
\pgfsetbuttcap
{
\definecolor{diafillcolor}{rgb}{0.000000, 0.000000, 0.000000}
\pgfsetfillcolor{diafillcolor}
\pgfsetfillopacity{1.000000}
\definecolor{dialinecolor}{rgb}{0.000000, 0.000000, 0.000000}
\pgfsetstrokecolor{dialinecolor}
\pgfsetstrokeopacity{1.000000}
\draw (31.951477\du,19.907776\du)--(19.548523\du,14.592224\du);
}
\pgfsetlinewidth{0.100000\du}
\pgfsetdash{}{0pt}
\pgfsetbuttcap
{
\definecolor{diafillcolor}{rgb}{0.000000, 0.000000, 0.000000}
\pgfsetfillcolor{diafillcolor}
\pgfsetfillopacity{1.000000}
\definecolor{dialinecolor}{rgb}{0.000000, 0.000000, 0.000000}
\pgfsetstrokecolor{dialinecolor}
\pgfsetstrokeopacity{1.000000}
\draw (63.949696\du,20.613590\du)--(19.550304\du,13.886410\du);
}
\pgfsetlinewidth{0.100000\du}
\pgfsetdash{}{0pt}
\pgfsetbuttcap
{
\definecolor{diafillcolor}{rgb}{0.000000, 0.000000, 0.000000}
\pgfsetfillcolor{diafillcolor}
\pgfsetfillopacity{1.000000}
\definecolor{dialinecolor}{rgb}{0.000000, 0.000000, 0.000000}
\pgfsetstrokecolor{dialinecolor}
\pgfsetstrokeopacity{1.000000}
\draw (79.953400\du,20.708405\du)--(19.546600\du,13.791595\du);
}
\pgfsetlinewidth{0.010000\du}
\pgfsetdash{}{0pt}
\pgfsetmiterjoin
\pgfsetbuttcap
{\pgfsetcornersarced{\pgfpoint{0.000000\du}{0.000000\du}}\definecolor{diafillcolor}{rgb}{1.000000, 1.000000, 1.000000}
\pgfsetfillcolor{diafillcolor}
\pgfsetfillopacity{1.000000}
\fill (32.000000\du,11.000000\du)--(32.000000\du,16.000000\du)--(37.000000\du,16.000000\du)--(37.000000\du,11.000000\du)--cycle;
}{\pgfsetcornersarced{\pgfpoint{0.000000\du}{0.000000\du}}\definecolor{dialinecolor}{rgb}{0.000000, 0.000000, 0.000000}
\pgfsetstrokecolor{dialinecolor}
\pgfsetstrokeopacity{1.000000}
\draw (32.000000\du,11.000000\du)--(32.000000\du,16.000000\du)--(37.000000\du,16.000000\du)--(37.000000\du,11.000000\du)--cycle;
}\pgfsetlinewidth{0.100000\du}
\pgfsetdash{}{0pt}
\pgfsetbuttcap
{
\definecolor{diafillcolor}{rgb}{0.000000, 0.000000, 0.000000}
\pgfsetfillcolor{diafillcolor}
\pgfsetfillopacity{1.000000}
\definecolor{dialinecolor}{rgb}{0.000000, 0.000000, 0.000000}
\pgfsetstrokecolor{dialinecolor}
\pgfsetstrokeopacity{1.000000}
\draw (63.951172\du,20.402618\du)--(37.048828\du,14.097382\du);
}
\pgfsetlinewidth{0.100000\du}
\pgfsetdash{}{0pt}
\pgfsetbuttcap
{
\definecolor{diafillcolor}{rgb}{0.000000, 0.000000, 0.000000}
\pgfsetfillcolor{diafillcolor}
\pgfsetfillopacity{1.000000}
\definecolor{dialinecolor}{rgb}{0.000000, 0.000000, 0.000000}
\pgfsetstrokecolor{dialinecolor}
\pgfsetstrokeopacity{1.000000}
\draw (79.950439\du,20.601631\du)--(37.049561\du,13.898369\du);
}
\pgfsetlinewidth{0.010000\du}
\pgfsetdash{}{0pt}
\pgfsetmiterjoin
\pgfsetbuttcap
{\pgfsetcornersarced{\pgfpoint{0.000000\du}{0.000000\du}}\definecolor{diafillcolor}{rgb}{1.000000, 1.000000, 1.000000}
\pgfsetfillcolor{diafillcolor}
\pgfsetfillopacity{1.000000}
\fill (64.000000\du,18.500000\du)--(64.000000\du,23.500000\du)--(69.000000\du,23.500000\du)--(69.000000\du,18.500000\du)--cycle;
}{\pgfsetcornersarced{\pgfpoint{0.000000\du}{0.000000\du}}\definecolor{dialinecolor}{rgb}{0.000000, 0.000000, 0.000000}
\pgfsetstrokecolor{dialinecolor}
\pgfsetstrokeopacity{1.000000}
\draw (64.000000\du,18.500000\du)--(64.000000\du,23.500000\du)--(69.000000\du,23.500000\du)--(69.000000\du,18.500000\du)--cycle;
}\pgfsetlinewidth{0.100000\du}
\pgfsetdash{}{0pt}
\pgfsetbuttcap
{
\definecolor{diafillcolor}{rgb}{0.000000, 0.000000, 0.000000}
\pgfsetfillcolor{diafillcolor}
\pgfsetfillopacity{1.000000}
\definecolor{dialinecolor}{rgb}{0.000000, 0.000000, 0.000000}
\pgfsetstrokecolor{dialinecolor}
\pgfsetstrokeopacity{1.000000}
\draw (79.951172\du,19.805237\du)--(69.048828\du,14.694763\du);
}
\pgfsetlinewidth{0.100000\du}
\pgfsetdash{}{0pt}
\pgfsetbuttcap
{
\definecolor{diafillcolor}{rgb}{0.000000, 0.000000, 0.000000}
\pgfsetfillcolor{diafillcolor}
\pgfsetfillopacity{1.000000}
\definecolor{dialinecolor}{rgb}{0.000000, 0.000000, 0.000000}
\pgfsetstrokecolor{dialinecolor}
\pgfsetstrokeopacity{1.000000}
\draw (31.951477\du,37.407776\du)--(19.548523\du,32.092224\du);
}
\pgfsetlinewidth{0.100000\du}
\pgfsetdash{}{0pt}
\pgfsetbuttcap
{
\definecolor{diafillcolor}{rgb}{0.000000, 0.000000, 0.000000}
\pgfsetfillcolor{diafillcolor}
\pgfsetfillopacity{1.000000}
\definecolor{dialinecolor}{rgb}{0.000000, 0.000000, 0.000000}
\pgfsetstrokecolor{dialinecolor}
\pgfsetstrokeopacity{1.000000}
\draw (47.949776\du,37.929054\du)--(19.550224\du,31.570946\du);
}
\pgfsetlinewidth{0.100000\du}
\pgfsetdash{}{0pt}
\pgfsetbuttcap
{
\definecolor{diafillcolor}{rgb}{0.000000, 0.000000, 0.000000}
\pgfsetfillcolor{diafillcolor}
\pgfsetfillopacity{1.000000}
\definecolor{dialinecolor}{rgb}{0.000000, 0.000000, 0.000000}
\pgfsetstrokecolor{dialinecolor}
\pgfsetstrokeopacity{1.000000}
\draw (79.953400\du,38.208405\du)--(19.546600\du,31.291595\du);
}
\pgfsetlinewidth{0.100000\du}
\pgfsetdash{}{0pt}
\pgfsetbuttcap
{
\definecolor{diafillcolor}{rgb}{0.000000, 0.000000, 0.000000}
\pgfsetfillcolor{diafillcolor}
\pgfsetfillopacity{1.000000}
\definecolor{dialinecolor}{rgb}{0.000000, 0.000000, 0.000000}
\pgfsetstrokecolor{dialinecolor}
\pgfsetstrokeopacity{1.000000}
\draw (47.951172\du,37.305237\du)--(37.048828\du,32.194763\du);
}
\pgfsetlinewidth{0.100000\du}
\pgfsetdash{}{0pt}
\pgfsetbuttcap
{
\definecolor{diafillcolor}{rgb}{0.000000, 0.000000, 0.000000}
\pgfsetfillcolor{diafillcolor}
\pgfsetfillopacity{1.000000}
\definecolor{dialinecolor}{rgb}{0.000000, 0.000000, 0.000000}
\pgfsetstrokecolor{dialinecolor}
\pgfsetstrokeopacity{1.000000}
\draw (79.950439\du,38.101631\du)--(37.049561\du,31.398369\du);
}
\pgfsetlinewidth{0.100000\du}
\pgfsetdash{}{0pt}
\pgfsetbuttcap
{
\definecolor{diafillcolor}{rgb}{0.000000, 0.000000, 0.000000}
\pgfsetfillcolor{diafillcolor}
\pgfsetfillopacity{1.000000}
\definecolor{dialinecolor}{rgb}{0.000000, 0.000000, 0.000000}
\pgfsetstrokecolor{dialinecolor}
\pgfsetstrokeopacity{1.000000}
\draw (79.951172\du,37.902618\du)--(53.048828\du,31.597382\du);
}
\pgfsetlinewidth{0.010000\du}
\pgfsetdash{}{0pt}
\pgfsetmiterjoin
\pgfsetbuttcap
{\pgfsetcornersarced{\pgfpoint{0.000000\du}{0.000000\du}}\definecolor{diafillcolor}{rgb}{1.000000, 1.000000, 1.000000}
\pgfsetfillcolor{diafillcolor}
\pgfsetfillopacity{1.000000}
\fill (48.000000\du,28.500000\du)--(48.000000\du,33.500000\du)--(53.000000\du,33.500000\du)--(53.000000\du,28.500000\du)--cycle;
}{\pgfsetcornersarced{\pgfpoint{0.000000\du}{0.000000\du}}\definecolor{dialinecolor}{rgb}{0.000000, 0.000000, 0.000000}
\pgfsetstrokecolor{dialinecolor}
\pgfsetstrokeopacity{1.000000}
\draw (48.000000\du,28.500000\du)--(48.000000\du,33.500000\du)--(53.000000\du,33.500000\du)--(53.000000\du,28.500000\du)--cycle;
}\pgfsetlinewidth{0.010000\du}
\pgfsetdash{}{0pt}
\pgfsetmiterjoin
\pgfsetbuttcap
{\pgfsetcornersarced{\pgfpoint{0.000000\du}{0.000000\du}}\definecolor{diafillcolor}{rgb}{1.000000, 1.000000, 1.000000}
\pgfsetfillcolor{diafillcolor}
\pgfsetfillopacity{1.000000}
\fill (32.000000\du,28.500000\du)--(32.000000\du,33.500000\du)--(37.000000\du,33.500000\du)--(37.000000\du,28.500000\du)--cycle;
}{\pgfsetcornersarced{\pgfpoint{0.000000\du}{0.000000\du}}\definecolor{dialinecolor}{rgb}{0.000000, 0.000000, 0.000000}
\pgfsetstrokecolor{dialinecolor}
\pgfsetstrokeopacity{1.000000}
\draw (32.000000\du,28.500000\du)--(32.000000\du,33.500000\du)--(37.000000\du,33.500000\du)--(37.000000\du,28.500000\du)--cycle;
}
\definecolor{dialinecolor}{rgb}{0.000000, 0.000000, 0.000000}
\pgfsetstrokecolor{dialinecolor}
\pgfsetstrokeopacity{1.000000}
\definecolor{diafillcolor}{rgb}{0.000000, 0.000000, 0.000000}
\pgfsetfillcolor{diafillcolor}
\pgfsetfillopacity{1.000000}
\node[anchor=base west,inner sep=0pt,outer sep=0pt,color=dialinecolor] at (92.000000\du,21.500000\du){...};
\node[anchor=base west,inner sep=0pt,outer sep=0pt,color=dialinecolor] at (15.000000\du,0.500000\du){$S_0$};
\node[anchor=base west,inner sep=0pt,outer sep=0pt,color=dialinecolor] at (33.000000\du,0.500000\du){$S_1$};
\node[anchor=base west,inner sep=0pt,outer sep=0pt,color=dialinecolor] at (48.000000\du,0.500000\du){$S_2$};
\node[anchor=base west,inner sep=0pt,outer sep=0pt,color=dialinecolor] at (65.000000\du,0.500000\du){$S_3$};
\node[anchor=base west,inner sep=0pt,outer sep=0pt,color=dialinecolor] at (80.000000\du,0.500000\du){$S_4$};
\node[anchor=base west,inner sep=0pt,outer sep=0pt,color=dialinecolor] at (9.000000\du,14.000000\du){$V_1^{0}$};
\node[anchor=base west,inner sep=0pt,outer sep=0pt,color=dialinecolor] at (9.000000\du,22.000000\du){$U_1^{0}$};
\node[anchor=base west,inner sep=0pt,outer sep=0pt,color=dialinecolor] at (9.000000\du,32.000000\du){$V_2^{0}$};
\node[anchor=base west,inner sep=0pt,outer sep=0pt,color=dialinecolor] at (9.000000\du,40.000000\du){$U_2^{0}$};
\node[anchor=base west,inner sep=0pt,outer sep=0pt,color=dialinecolor] at (86.000000\du,14.000000\du){$V_1^{4}$};
\node[anchor=base west,inner sep=0pt,outer sep=0pt,color=dialinecolor] at (86.000000\du,22.000000\du){$U_1^{4}$};
\node[anchor=base west,inner sep=0pt,outer sep=0pt,color=dialinecolor] at (86.000000\du,32.000000\du){$V_2^{4}$};
\node[anchor=base west,inner sep=0pt,outer sep=0pt,color=dialinecolor] at (86.000000\du,40.000000\du){$U_2^{4}$};
\definecolor{dialinecolor}{rgb}{0.000000, 0.000000, 0.000000}
\pgfsetstrokecolor{dialinecolor}
\pgfsetstrokeopacity{1.000000}
\definecolor{diafillcolor}{rgb}{0.000000, 0.000000, 0.000000}
\pgfsetfillcolor{diafillcolor}
\pgfsetfillopacity{1.000000}
\node[anchor=base west,inner sep=0pt,outer sep=0pt,color=dialinecolor] at (56.000000\du,46.500000\du){...};
\end{tikzpicture}

%% file: figures/red2.tex
\ifx\du\undefined
  \newlength{\du}
\fi
\setlength{\du}{15\unitlength}
\begin{tikzpicture}[even odd rule,scale=0.40]
\pgftransformxscale{1.000000}
\pgftransformyscale{-1.000000}
\definecolor{dialinecolor}{rgb}{0.000000, 0.000000, 0.000000}
\pgfsetstrokecolor{dialinecolor}
\pgfsetstrokeopacity{1.000000}
\definecolor{diafillcolor}{rgb}{1.000000, 1.000000, 1.000000}
\pgfsetfillcolor{diafillcolor}
\pgfsetfillopacity{1.000000}
\pgfsetlinewidth{0.0500000\du}
\pgfsetdash{}{0pt}
\pgfsetmiterjoin
\pgfsetbuttcap
{\pgfsetcornersarced{\pgfpoint{0.000000\du}{0.000000\du}}\definecolor{diafillcolor}{rgb}{1.000000, 1.000000, 1.000000}
\pgfsetfillcolor{diafillcolor}
\pgfsetfillopacity{1.000000}
\fill (14.500000\du,11.000000\du)--(14.500000\du,16.000000\du)--(19.500000\du,16.000000\du)--(19.500000\du,11.000000\du)--cycle;
}{\pgfsetcornersarced{\pgfpoint{0.000000\du}{0.000000\du}}\definecolor{dialinecolor}{rgb}{0.000000, 0.000000, 0.000000}
\pgfsetstrokecolor{dialinecolor}
\pgfsetstrokeopacity{1.000000}
\draw (14.500000\du,11.000000\du)--(14.500000\du,16.000000\du)--(19.500000\du,16.000000\du)--(19.500000\du,11.000000\du)--cycle;
}\pgfsetlinewidth{0.0500000\du}
\pgfsetdash{}{0pt}
\pgfsetbuttcap
{
\definecolor{diafillcolor}{rgb}{0.000000, 0.000000, 0.000000}
\pgfsetfillcolor{diafillcolor}
\pgfsetfillopacity{1.000000}
\definecolor{dialinecolor}{rgb}{0.000000, 0.000000, 0.000000}
\pgfsetstrokecolor{dialinecolor}
\pgfsetstrokeopacity{1.000000}
\draw (14.500000\du,13.500000\du)--(19.500000\du,13.500000\du);
}
\pgfsetlinewidth{0.0500000\du}
\pgfsetdash{}{0pt}
\pgfsetbuttcap
{
\definecolor{diafillcolor}{rgb}{0.000000, 0.000000, 0.000000}
\pgfsetfillcolor{diafillcolor}
\pgfsetfillopacity{1.000000}
\definecolor{dialinecolor}{rgb}{0.000000, 0.000000, 0.000000}
\pgfsetstrokecolor{dialinecolor}
\pgfsetstrokeopacity{1.000000}
\draw (17.000000\du,11.000000\du)--(17.000000\du,16.000000\du);
}
\pgfsetlinewidth{0.0500000\du}
\pgfsetdash{}{0pt}
\pgfsetmiterjoin
\pgfsetbuttcap
{\pgfsetcornersarced{\pgfpoint{0.000000\du}{0.000000\du}}\definecolor{diafillcolor}{rgb}{1.000000, 1.000000, 1.000000}
\pgfsetfillcolor{diafillcolor}
\pgfsetfillopacity{1.000000}
\fill (14.500000\du,18.500000\du)--(14.500000\du,23.500000\du)--(19.500000\du,23.500000\du)--(19.500000\du,18.500000\du)--cycle;
}{\pgfsetcornersarced{\pgfpoint{0.000000\du}{0.000000\du}}\definecolor{dialinecolor}{rgb}{0.000000, 0.000000, 0.000000}
\pgfsetstrokecolor{dialinecolor}
\pgfsetstrokeopacity{1.000000}
\draw (14.500000\du,18.500000\du)--(14.500000\du,23.500000\du)--(19.500000\du,23.500000\du)--(19.500000\du,18.500000\du)--cycle;
}\pgfsetlinewidth{0.0500000\du}
\pgfsetdash{}{0pt}
\pgfsetbuttcap
{
\definecolor{diafillcolor}{rgb}{0.000000, 0.000000, 0.000000}
\pgfsetfillcolor{diafillcolor}
\pgfsetfillopacity{1.000000}
\definecolor{dialinecolor}{rgb}{0.000000, 0.000000, 0.000000}
\pgfsetstrokecolor{dialinecolor}
\pgfsetstrokeopacity{1.000000}
\draw (14.500000\du,21.000000\du)--(19.500000\du,21.000000\du);
}
\pgfsetlinewidth{0.0500000\du}
\pgfsetdash{}{0pt}
\pgfsetbuttcap
{
\definecolor{diafillcolor}{rgb}{0.000000, 0.000000, 0.000000}
\pgfsetfillcolor{diafillcolor}
\pgfsetfillopacity{1.000000}
\definecolor{dialinecolor}{rgb}{0.000000, 0.000000, 0.000000}
\pgfsetstrokecolor{dialinecolor}
\pgfsetstrokeopacity{1.000000}
\draw (17.000000\du,18.500000\du)--(17.000000\du,23.500000\du);
}
\pgfsetlinewidth{0.0500000\du}
\pgfsetdash{}{0pt}
\definecolor{diafillcolor}{rgb}{1.000000, 1.000000, 1.000000}
\pgfsetfillcolor{diafillcolor}
\pgfsetfillopacity{1.000000}
\pgfpathellipse{\pgfpoint{17.000000\du}{5.000000\du}}{\pgfpoint{5.000000\du}{0\du}}{\pgfpoint{0\du}{3.000000\du}}
\pgfusepath{fill}
\definecolor{dialinecolor}{rgb}{0.000000, 0.000000, 0.000000}
\pgfsetstrokecolor{dialinecolor}
\pgfsetstrokeopacity{1.000000}
\pgfpathellipse{\pgfpoint{17.000000\du}{5.000000\du}}{\pgfpoint{5.000000\du}{0\du}}{\pgfpoint{0\du}{3.000000\du}}
\pgfusepath{stroke}
\pgfsetlinewidth{0.0500000\du}
\pgfsetdash{}{0pt}
\definecolor{diafillcolor}{rgb}{0.000000, 0.000000, 0.000000}
\pgfsetfillcolor{diafillcolor}
\pgfsetfillopacity{1.000000}
\pgfpathellipse{\pgfpoint{14.000000\du}{5.000000\du}}{\pgfpoint{0.500000\du}{0\du}}{\pgfpoint{0\du}{0.500000\du}}
\pgfusepath{fill}
\definecolor{dialinecolor}{rgb}{0.000000, 0.000000, 0.000000}
\pgfsetstrokecolor{dialinecolor}
\pgfsetstrokeopacity{1.000000}
\pgfpathellipse{\pgfpoint{14.000000\du}{5.000000\du}}{\pgfpoint{0.500000\du}{0\du}}{\pgfpoint{0\du}{0.500000\du}}
\pgfusepath{stroke}
\pgfsetlinewidth{0.0500000\du}
\pgfsetdash{}{0pt}
\definecolor{diafillcolor}{rgb}{0.000000, 0.000000, 0.000000}
\pgfsetfillcolor{diafillcolor}
\pgfsetfillopacity{1.000000}
\pgfpathellipse{\pgfpoint{16.000000\du}{5.000000\du}}{\pgfpoint{0.500000\du}{0\du}}{\pgfpoint{0\du}{0.500000\du}}
\pgfusepath{fill}
\definecolor{dialinecolor}{rgb}{0.000000, 0.000000, 0.000000}
\pgfsetstrokecolor{dialinecolor}
\pgfsetstrokeopacity{1.000000}
\pgfpathellipse{\pgfpoint{16.000000\du}{5.000000\du}}{\pgfpoint{0.500000\du}{0\du}}{\pgfpoint{0\du}{0.500000\du}}
\pgfusepath{stroke}
\pgfsetlinewidth{0.0500000\du}
\pgfsetdash{}{0pt}
\definecolor{diafillcolor}{rgb}{0.000000, 0.000000, 0.000000}
\pgfsetfillcolor{diafillcolor}
\pgfsetfillopacity{1.000000}
\pgfpathellipse{\pgfpoint{18.000000\du}{5.000000\du}}{\pgfpoint{0.500000\du}{0\du}}{\pgfpoint{0\du}{0.500000\du}}
\pgfusepath{fill}
\definecolor{dialinecolor}{rgb}{0.000000, 0.000000, 0.000000}
\pgfsetstrokecolor{dialinecolor}
\pgfsetstrokeopacity{1.000000}
\pgfpathellipse{\pgfpoint{18.000000\du}{5.000000\du}}{\pgfpoint{0.500000\du}{0\du}}{\pgfpoint{0\du}{0.500000\du}}
\pgfusepath{stroke}
\pgfsetlinewidth{0.0500000\du}
\pgfsetdash{}{0pt}
\definecolor{diafillcolor}{rgb}{0.000000, 0.000000, 0.000000}
\pgfsetfillcolor{diafillcolor}
\pgfsetfillopacity{1.000000}
\pgfpathellipse{\pgfpoint{20.000000\du}{5.000000\du}}{\pgfpoint{0.500000\du}{0\du}}{\pgfpoint{0\du}{0.500000\du}}
\pgfusepath{fill}
\definecolor{dialinecolor}{rgb}{0.000000, 0.000000, 0.000000}
\pgfsetstrokecolor{dialinecolor}
\pgfsetstrokeopacity{1.000000}
\pgfpathellipse{\pgfpoint{20.000000\du}{5.000000\du}}{\pgfpoint{0.500000\du}{0\du}}{\pgfpoint{0\du}{0.500000\du}}
\pgfusepath{stroke}
\pgfsetlinewidth{0.0500000\du}
\pgfsetdash{}{0pt}
\definecolor{diafillcolor}{rgb}{0.000000, 0.000000, 0.000000}
\pgfsetfillcolor{diafillcolor}
\pgfsetfillopacity{1.000000}
\pgfpathellipse{\pgfpoint{15.701031\du}{12.191443\du}}{\pgfpoint{0.500000\du}{0\du}}{\pgfpoint{0\du}{0.500000\du}}
\pgfusepath{fill}
\definecolor{dialinecolor}{rgb}{0.000000, 0.000000, 0.000000}
\pgfsetstrokecolor{dialinecolor}
\pgfsetstrokeopacity{1.000000}
\pgfpathellipse{\pgfpoint{15.701031\du}{12.191443\du}}{\pgfpoint{0.500000\du}{0\du}}{\pgfpoint{0\du}{0.500000\du}}
\pgfusepath{stroke}
\pgfsetlinewidth{0.0500000\du}
\pgfsetdash{}{0pt}
\definecolor{diafillcolor}{rgb}{0.000000, 0.000000, 0.000000}
\pgfsetfillcolor{diafillcolor}
\pgfsetfillopacity{1.000000}
\pgfpathellipse{\pgfpoint{17.623505\du}{11.671856\du}}{\pgfpoint{0.500000\du}{0\du}}{\pgfpoint{0\du}{0.500000\du}}
\pgfusepath{fill}
\definecolor{dialinecolor}{rgb}{0.000000, 0.000000, 0.000000}
\pgfsetstrokecolor{dialinecolor}
\pgfsetstrokeopacity{1.000000}
\pgfpathellipse{\pgfpoint{17.623505\du}{11.671856\du}}{\pgfpoint{0.500000\du}{0\du}}{\pgfpoint{0\du}{0.500000\du}}
\pgfusepath{stroke}
\pgfsetlinewidth{0.0500000\du}
\pgfsetdash{}{0pt}
\definecolor{diafillcolor}{rgb}{0.000000, 0.000000, 0.000000}
\pgfsetfillcolor{diafillcolor}
\pgfsetfillopacity{1.000000}
\pgfpathellipse{\pgfpoint{18.636701\du}{12.685052\du}}{\pgfpoint{0.500000\du}{0\du}}{\pgfpoint{0\du}{0.500000\du}}
\pgfusepath{fill}
\definecolor{dialinecolor}{rgb}{0.000000, 0.000000, 0.000000}
\pgfsetstrokecolor{dialinecolor}
\pgfsetstrokeopacity{1.000000}
\pgfpathellipse{\pgfpoint{18.636701\du}{12.685052\du}}{\pgfpoint{0.500000\du}{0\du}}{\pgfpoint{0\du}{0.500000\du}}
\pgfusepath{stroke}
\pgfsetlinewidth{0.0500000\du}
\pgfsetdash{}{0pt}
\definecolor{diafillcolor}{rgb}{0.000000, 0.000000, 0.000000}
\pgfsetfillcolor{diafillcolor}
\pgfsetfillopacity{1.000000}
\pgfpathellipse{\pgfpoint{15.727010\du}{14.659485\du}}{\pgfpoint{0.500000\du}{0\du}}{\pgfpoint{0\du}{0.500000\du}}
\pgfusepath{fill}
\definecolor{dialinecolor}{rgb}{0.000000, 0.000000, 0.000000}
\pgfsetstrokecolor{dialinecolor}
\pgfsetstrokeopacity{1.000000}
\pgfpathellipse{\pgfpoint{15.727010\du}{14.659485\du}}{\pgfpoint{0.500000\du}{0\du}}{\pgfpoint{0\du}{0.500000\du}}
\pgfusepath{stroke}
\pgfsetlinewidth{0.0500000\du}
\pgfsetdash{}{0pt}
\definecolor{diafillcolor}{rgb}{0.000000, 0.000000, 0.000000}
\pgfsetfillcolor{diafillcolor}
\pgfsetfillopacity{1.000000}
\pgfpathellipse{\pgfpoint{17.711072\du}{14.204082\du}}{\pgfpoint{0.500000\du}{0\du}}{\pgfpoint{0\du}{0.500000\du}}
\pgfusepath{fill}
\definecolor{dialinecolor}{rgb}{0.000000, 0.000000, 0.000000}
\pgfsetstrokecolor{dialinecolor}
\pgfsetstrokeopacity{1.000000}
\pgfpathellipse{\pgfpoint{17.711072\du}{14.204082\du}}{\pgfpoint{0.500000\du}{0\du}}{\pgfpoint{0\du}{0.500000\du}}
\pgfusepath{stroke}
\pgfsetlinewidth{0.0500000\du}
\pgfsetdash{}{0pt}
\definecolor{diafillcolor}{rgb}{0.000000, 0.000000, 0.000000}
\pgfsetfillcolor{diafillcolor}
\pgfsetfillopacity{1.000000}
\pgfpathellipse{\pgfpoint{18.724268\du}{15.217278\du}}{\pgfpoint{0.500000\du}{0\du}}{\pgfpoint{0\du}{0.500000\du}}
\pgfusepath{fill}
\definecolor{dialinecolor}{rgb}{0.000000, 0.000000, 0.000000}
\pgfsetstrokecolor{dialinecolor}
\pgfsetstrokeopacity{1.000000}
\pgfpathellipse{\pgfpoint{18.724268\du}{15.217278\du}}{\pgfpoint{0.500000\du}{0\du}}{\pgfpoint{0\du}{0.500000\du}}
\pgfusepath{stroke}
\pgfsetlinewidth{0.0500000\du}
\pgfsetdash{}{0pt}
\definecolor{diafillcolor}{rgb}{0.000000, 0.000000, 0.000000}
\pgfsetfillcolor{diafillcolor}
\pgfsetfillopacity{1.000000}
\pgfpathellipse{\pgfpoint{15.191072\du}{19.296041\du}}{\pgfpoint{0.500000\du}{0\du}}{\pgfpoint{0\du}{0.500000\du}}
\pgfusepath{fill}
\definecolor{dialinecolor}{rgb}{0.000000, 0.000000, 0.000000}
\pgfsetstrokecolor{dialinecolor}
\pgfsetstrokeopacity{1.000000}
\pgfpathellipse{\pgfpoint{15.191072\du}{19.296041\du}}{\pgfpoint{0.500000\du}{0\du}}{\pgfpoint{0\du}{0.500000\du}}
\pgfusepath{stroke}
\pgfsetlinewidth{0.0500000\du}
\pgfsetdash{}{0pt}
\definecolor{diafillcolor}{rgb}{0.000000, 0.000000, 0.000000}
\pgfsetfillcolor{diafillcolor}
\pgfsetfillopacity{1.000000}
\pgfpathellipse{\pgfpoint{16.204268\du}{20.309237\du}}{\pgfpoint{0.500000\du}{0\du}}{\pgfpoint{0\du}{0.500000\du}}
\pgfusepath{fill}
\definecolor{dialinecolor}{rgb}{0.000000, 0.000000, 0.000000}
\pgfsetstrokecolor{dialinecolor}
\pgfsetstrokeopacity{1.000000}
\pgfpathellipse{\pgfpoint{16.204268\du}{20.309237\du}}{\pgfpoint{0.500000\du}{0\du}}{\pgfpoint{0\du}{0.500000\du}}
\pgfusepath{stroke}
\pgfsetlinewidth{0.0500000\du}
\pgfsetdash{}{0pt}
\definecolor{diafillcolor}{rgb}{0.000000, 0.000000, 0.000000}
\pgfsetfillcolor{diafillcolor}
\pgfsetfillopacity{1.000000}
\pgfpathellipse{\pgfpoint{15.278639\du}{21.828268\du}}{\pgfpoint{0.500000\du}{0\du}}{\pgfpoint{0\du}{0.500000\du}}
\pgfusepath{fill}
\definecolor{dialinecolor}{rgb}{0.000000, 0.000000, 0.000000}
\pgfsetstrokecolor{dialinecolor}
\pgfsetstrokeopacity{1.000000}
\pgfpathellipse{\pgfpoint{15.278639\du}{21.828268\du}}{\pgfpoint{0.500000\du}{0\du}}{\pgfpoint{0\du}{0.500000\du}}
\pgfusepath{stroke}
\pgfsetlinewidth{0.0500000\du}
\pgfsetdash{}{0pt}
\definecolor{diafillcolor}{rgb}{0.000000, 0.000000, 0.000000}
\pgfsetfillcolor{diafillcolor}
\pgfsetfillopacity{1.000000}
\pgfpathellipse{\pgfpoint{16.291835\du}{22.841464\du}}{\pgfpoint{0.500000\du}{0\du}}{\pgfpoint{0\du}{0.500000\du}}
\pgfusepath{fill}
\definecolor{dialinecolor}{rgb}{0.000000, 0.000000, 0.000000}
\pgfsetstrokecolor{dialinecolor}
\pgfsetstrokeopacity{1.000000}
\pgfpathellipse{\pgfpoint{16.291835\du}{22.841464\du}}{\pgfpoint{0.500000\du}{0\du}}{\pgfpoint{0\du}{0.500000\du}}
\pgfusepath{stroke}
\pgfsetlinewidth{0.0500000\du}
\pgfsetdash{}{0pt}
\definecolor{diafillcolor}{rgb}{0.000000, 0.000000, 0.000000}
\pgfsetfillcolor{diafillcolor}
\pgfsetfillopacity{1.000000}
\pgfpathellipse{\pgfpoint{18.230660\du}{19.659753\du}}{\pgfpoint{0.500000\du}{0\du}}{\pgfpoint{0\du}{0.500000\du}}
\pgfusepath{fill}
\definecolor{dialinecolor}{rgb}{0.000000, 0.000000, 0.000000}
\pgfsetstrokecolor{dialinecolor}
\pgfsetstrokeopacity{1.000000}
\pgfpathellipse{\pgfpoint{18.230660\du}{19.659753\du}}{\pgfpoint{0.500000\du}{0\du}}{\pgfpoint{0\du}{0.500000\du}}
\pgfusepath{stroke}
\pgfsetlinewidth{0.0500000\du}
\pgfsetdash{}{0pt}
\definecolor{diafillcolor}{rgb}{0.000000, 0.000000, 0.000000}
\pgfsetfillcolor{diafillcolor}
\pgfsetfillopacity{1.000000}
\pgfpathellipse{\pgfpoint{18.256639\du}{22.127794\du}}{\pgfpoint{0.500000\du}{0\du}}{\pgfpoint{0\du}{0.500000\du}}
\pgfusepath{fill}
\definecolor{dialinecolor}{rgb}{0.000000, 0.000000, 0.000000}
\pgfsetstrokecolor{dialinecolor}
\pgfsetstrokeopacity{1.000000}
\pgfpathellipse{\pgfpoint{18.256639\du}{22.127794\du}}{\pgfpoint{0.500000\du}{0\du}}{\pgfpoint{0\du}{0.500000\du}}
\pgfusepath{stroke}
\pgfsetlinewidth{0.0500000\du}
\pgfsetdash{{\pgflinewidth}{0.100000\du}}{0cm}
\pgfsetbuttcap
{
\definecolor{diafillcolor}{rgb}{0.000000, 0.000000, 0.000000}
\pgfsetfillcolor{diafillcolor}
\pgfsetfillopacity{1.000000}
\pgfsetarrowsend{stealth}
\definecolor{dialinecolor}{rgb}{0.000000, 0.000000, 0.000000}
\pgfsetstrokecolor{dialinecolor}
\pgfsetstrokeopacity{1.000000}
\pgfpathmoveto{\pgfpoint{13.500254\du}{4.999556\du}}
\pgfpatharc{210}{137}{11.631481\du and 11.631481\du}
\pgfusepath{stroke}
}
\pgfsetlinewidth{0.0500000\du}
\pgfsetdash{{\pgflinewidth}{0.100000\du}}{0cm}
\pgfsetbuttcap
{
\definecolor{diafillcolor}{rgb}{0.000000, 0.000000, 0.000000}
\pgfsetfillcolor{diafillcolor}
\pgfsetfillopacity{1.000000}
\pgfsetarrowsend{stealth}
\definecolor{dialinecolor}{rgb}{0.000000, 0.000000, 0.000000}
\pgfsetstrokecolor{dialinecolor}
\pgfsetstrokeopacity{1.000000}
\pgfpathmoveto{\pgfpoint{13.500123\du}{4.998528\du}}
\pgfpatharc{185}{155}{28.828321\du and 28.828321\du}
\pgfusepath{stroke}
}
\pgfsetlinewidth{0.0500000\du}
\pgfsetdash{{\pgflinewidth}{0.100000\du}}{0cm}
\pgfsetbuttcap
{
\definecolor{diafillcolor}{rgb}{0.000000, 0.000000, 0.000000}
\pgfsetfillcolor{diafillcolor}
\pgfsetfillopacity{1.000000}
\pgfsetarrowsend{stealth}
\definecolor{dialinecolor}{rgb}{0.000000, 0.000000, 0.000000}
\pgfsetstrokecolor{dialinecolor}
\pgfsetstrokeopacity{1.000000}
\pgfpathmoveto{\pgfpoint{13.500020\du}{4.999474\du}}
\pgfpatharc{183}{142}{10.110589\du and 10.110589\du}
\pgfusepath{stroke}
}
\definecolor{dialinecolor}{rgb}{0.000000, 0.000000, 0.000000}
\pgfsetstrokecolor{dialinecolor}
\pgfsetstrokeopacity{1.000000}
\definecolor{diafillcolor}{rgb}{0.000000, 0.000000, 0.000000}
\pgfsetfillcolor{diafillcolor}
\pgfsetfillopacity{1.000000}
\node[anchor=base west,inner sep=0pt,outer sep=0pt,color=dialinecolor] at (13.674639\du,4.168660\du){$a_1$};
\definecolor{dialinecolor}{rgb}{0.000000, 0.000000, 0.000000}
\pgfsetstrokecolor{dialinecolor}
\pgfsetstrokeopacity{1.000000}
\definecolor{diafillcolor}{rgb}{0.000000, 0.000000, 0.000000}
\pgfsetfillcolor{diafillcolor}
\pgfsetfillopacity{1.000000}
\node[anchor=base west,inner sep=0pt,outer sep=0pt,color=dialinecolor] at (15.654268\du,4.091144\du){$a_2$};
\definecolor{dialinecolor}{rgb}{0.000000, 0.000000, 0.000000}
\pgfsetstrokecolor{dialinecolor}
\pgfsetstrokeopacity{1.000000}
\definecolor{diafillcolor}{rgb}{0.000000, 0.000000, 0.000000}
\pgfsetfillcolor{diafillcolor}
\pgfsetfillopacity{1.000000}
\node[anchor=base west,inner sep=0pt,outer sep=0pt,color=dialinecolor] at (17.550763\du,4.091144\du){$a_3$};
\definecolor{dialinecolor}{rgb}{0.000000, 0.000000, 0.000000}
\pgfsetstrokecolor{dialinecolor}
\pgfsetstrokeopacity{1.000000}
\definecolor{diafillcolor}{rgb}{0.000000, 0.000000, 0.000000}
\pgfsetfillcolor{diafillcolor}
\pgfsetfillopacity{1.000000}
\node[anchor=base west,inner sep=0pt,outer sep=0pt,color=dialinecolor] at (19.629113\du,4.143103\du){$a_4$};
\definecolor{dialinecolor}{rgb}{0.000000, 0.000000, 0.000000}
\pgfsetstrokecolor{dialinecolor}
\pgfsetstrokeopacity{1.000000}
\definecolor{diafillcolor}{rgb}{0.000000, 0.000000, 0.000000}
\pgfsetfillcolor{diafillcolor}
\pgfsetfillopacity{1.000000}
\node[anchor=base west,inner sep=0pt,outer sep=0pt,color=dialinecolor] at (14.952825\du,10.551144\du){$V_1^{0,a_1}$};
\definecolor{dialinecolor}{rgb}{0.000000, 0.000000, 0.000000}
\pgfsetstrokecolor{dialinecolor}
\pgfsetstrokeopacity{1.000000}
\definecolor{diafillcolor}{rgb}{0.000000, 0.000000, 0.000000}
\pgfsetfillcolor{diafillcolor}
\pgfsetfillopacity{1.000000}
\node[anchor=base west,inner sep=0pt,outer sep=0pt,color=dialinecolor] at (19.784990\du,11.703103\du){$V_1^{0,a_2}$};
\definecolor{dialinecolor}{rgb}{0.000000, 0.000000, 0.000000}
\pgfsetstrokecolor{dialinecolor}
\pgfsetstrokeopacity{1.000000}
\definecolor{diafillcolor}{rgb}{0.000000, 0.000000, 0.000000}
\pgfsetfillcolor{diafillcolor}
\pgfsetfillopacity{1.000000}
\node[anchor=base west,inner sep=0pt,outer sep=0pt,color=dialinecolor] at (11.032825\du,14.275062\du){$V_1^{0,a_3}$};
\definecolor{dialinecolor}{rgb}{0.000000, 0.000000, 0.000000}
\pgfsetstrokecolor{dialinecolor}
\pgfsetstrokeopacity{1.000000}
\definecolor{diafillcolor}{rgb}{0.000000, 0.000000, 0.000000}
\pgfsetfillcolor{diafillcolor}
\pgfsetfillopacity{1.000000}
\node[anchor=base west,inner sep=0pt,outer sep=0pt,color=dialinecolor] at (19.629113\du,14.301041\du){$V_1^{0,a_4}$};
\definecolor{dialinecolor}{rgb}{0.000000, 0.000000, 0.000000}
\pgfsetstrokecolor{dialinecolor}
\pgfsetstrokeopacity{1.000000}
\definecolor{diafillcolor}{rgb}{0.000000, 0.000000, 0.000000}
\pgfsetfillcolor{diafillcolor}
\pgfsetfillopacity{1.000000}
\node[anchor=base west,inner sep=0pt,outer sep=0pt,color=dialinecolor] at (19.759010\du,22.664959\du){$U_1^{0,a_4}$};
\definecolor{dialinecolor}{rgb}{0.000000, 0.000000, 0.000000}
\pgfsetstrokecolor{dialinecolor}
\pgfsetstrokeopacity{1.000000}
\definecolor{diafillcolor}{rgb}{0.000000, 0.000000, 0.000000}
\pgfsetfillcolor{diafillcolor}
\pgfsetfillopacity{1.000000}
\node[anchor=base west,inner sep=0pt,outer sep=0pt,color=dialinecolor] at (19.707052\du,19.941041\du){$U_1^{0,a_2}$};
\definecolor{dialinecolor}{rgb}{0.000000, 0.000000, 0.000000}
\pgfsetstrokecolor{dialinecolor}
\pgfsetstrokeopacity{1.000000}
\definecolor{diafillcolor}{rgb}{0.000000, 0.000000, 0.000000}
\pgfsetfillcolor{diafillcolor}
\pgfsetfillopacity{1.000000}
\node[anchor=base west,inner sep=0pt,outer sep=0pt,color=dialinecolor] at (11.458804\du,19.959186\du){$U_1^{0,a_1}$};
\definecolor{dialinecolor}{rgb}{0.000000, 0.000000, 0.000000}
\pgfsetstrokecolor{dialinecolor}
\pgfsetstrokeopacity{1.000000}
\definecolor{diafillcolor}{rgb}{0.000000, 0.000000, 0.000000}
\pgfsetfillcolor{diafillcolor}
\pgfsetfillopacity{1.000000}
\node[anchor=base west,inner sep=0pt,outer sep=0pt,color=dialinecolor] at (11.328907\du,22.601247\du){$U_1^{0,a_3}$};
\definecolor{dialinecolor}{rgb}{0.000000, 0.000000, 0.000000}
\pgfsetstrokecolor{dialinecolor}
\pgfsetstrokeopacity{1.000000}
\definecolor{diafillcolor}{rgb}{0.000000, 0.000000, 0.000000}
\pgfsetfillcolor{diafillcolor}
\pgfsetfillopacity{1.000000}
\node[anchor=base west,inner sep=0pt,outer sep=0pt,color=dialinecolor] at (14.272165\du,8.788144\du){};
\definecolor{dialinecolor}{rgb}{0.000000, 0.000000, 0.000000}
\pgfsetstrokecolor{dialinecolor}
\pgfsetstrokeopacity{1.000000}
\definecolor{diafillcolor}{rgb}{0.000000, 0.000000, 0.000000}
\pgfsetfillcolor{diafillcolor}
\pgfsetfillopacity{1.000000}
\node[anchor=base west,inner sep=0pt,outer sep=0pt,color=dialinecolor] at (10.484371\du,7.780216\du){};
\definecolor{dialinecolor}{rgb}{0.000000, 0.000000, 0.000000}
\pgfsetstrokecolor{dialinecolor}
\pgfsetstrokeopacity{1.000000}
\definecolor{diafillcolor}{rgb}{0.000000, 0.000000, 0.000000}
\pgfsetfillcolor{diafillcolor}
\pgfsetfillopacity{1.000000}
\node[anchor=base west,inner sep=0pt,outer sep=0pt,color=dialinecolor] at (12.276948\du,9.338979\du){};
\end{tikzpicture}

%% file: mpw.tex
We present a lower bound on the complexity of $k$-\kc
parameterized by modular pathwidth. Specifically, we show that, under the SETH,
no algorithm can solve $k$-\kc in $O^*\left( ( {k\choose \lfloor
k/2\rfloor}-\epsilon)^{\mpw} \right)$. This bound is somewhat weaker than the
one in Theorem \ref{thm:cw}, which is natural since modular pathwidth is more
restrictive than clique-width. As we see in Section \ref{sec:algs}, however,
this bound is tight. 
We complete
this section by performing the same reduction with the ETH, rather than the
SETH, as a starting point.  Under this weaker assumption we prove that $k$-\kc
does not admit an algorithm running in $n^{o(\mpw)}$, even when $k=O(\log n)$,
which implies that the problem does not admit an algorithm running in
$2^{o(k\cdot\mpw)}$. This is tight, and also applies to clique-width (Lemma
\ref{lem:parameters}). In this way, our reduction gives an alternative proof
that $k$-\kc is unlikely to be FPT parameterized by clique-width, even in
instances with a logarithmic number of colors.

\subsection{SETH-based Lower Bound}

\begin{theorem}\label{thm:mpw}

For any $k\ge3, \epsilon>0$, if there exists an algorithm solving $k$-coloring
in time $O^*\left(({k\choose \lfloor k/2\rfloor}-\epsilon)^{\mpw}\right)$,
where $\mpw$ is the graph's modular pathwidth, then the SETH is false.

\end{theorem}

As in Theorem \ref{thm:cw}, we begin our reduction from a \qcsp instance, where
the alphabet size $B$ is equal to the base of the exponential in our lower
bound. The intuition will be that the ``important'' vertices of the bags in a
modular tree decomposition of our graph will correspond to classes of $\lfloor
k/2\rfloor$ true twin vertices.  The set of $\lfloor k/2\rfloor$ colors used to
color them will encode the value of one variable of the original instance. 


\subparagraph*{Construction.} We are given some $k\ge3, \epsilon>0$. Let
$B={k\choose \lfloor k/2\rfloor}$.  Let $q$ be the smallest integer such that
$n$-variable \qcsp does not admit an $O^*\left((B-\epsilon)^n\right)$
algorithm.  Consider an arbitrary $n$-variable instance of \qcsp, call it
$\phi$.  We use the existence of the supposed algorithm to obtain an
$O^*\left((B-\epsilon)^n\right)$ algorithm that decides $\phi$, contradicting
the SETH.

As in Theorem \ref{thm:cw} we define a one-to-one translation function $T$.
This time, when $T$ is given as input a value $v\in\{1,\ldots,B\}$, it returns
a subset of $\{1,\ldots,k\}$ of cardinality $\lfloor k/2\rfloor$.  Let
$X=\{x_1,\ldots,x_n\}$ be the set of the $n$ variables of the \qcsp instance
and $C=\{c_0,\ldots,c_{m-1}\}$ the set of its $m$ constraints. We construct our
graph $G(\phi)$ as follows, where if we don't specify a list for a vertex its
list is $\{1,\ldots,k\}$:

\begin{enumerate}

\item For each variable $i\in\{1,\ldots,n\}$ we construct a clique $V_i$ on
$\lfloor k/2\rfloor$ vertices.

\item For each $j\in\{0,\ldots,m-1\}$, let $S$ be the set of satisfying
assignments of the constraint $c_j$. We construct an independent set $S_j$ with
one vertex for each element of $S$. We construct an OR($S_j$) gadget on the set
$S_j$.

\item For each $j\in\{0,\ldots,m-1\}$, each satisfying assignment $a$ for the
constraint $c_j$, and each variable $x_i$ appearing in $c_j$ we do the
following:

	\begin{enumerate}

	\item We construct a set $U_i^{j,a}$ of $\lceil k/2\rceil$ vertices.

	\item Suppose $a$ assigns value $v$ to $x_i$. For each color
$c\in\{1,\ldots,k\}\setminus T(v)$ we select a vertex of $U_i^{j,a}$. We
construct a $(1\to c)$-implication gadget from the vertex representing $a$ in
$S_j$ to this vertex of $U_i^{j,a}$.

	\item We connect all vertices of $U_i^{j,a}$ with all vertices of
$V_i$.

	\end{enumerate}

\end{enumerate}


\begin{lemma}\label{lem:mpw1}

	If $\phi$ is a satisfiable \qcsp instance, then $G(\phi)$ admits a
	proper list coloring.

\end{lemma}

\begin{proof}

Suppose we have a satisfying assignment to $\phi$. If the assignment gives
value $v_i$ to variable $x_i$, we use the colors of $T(v_i)$ to color $V_i$, in
some arbitrary way. For each constraint $c_j$, the assignment gives values to
the variables of $c_j$ consistent with some satisfying assignment $a$ of the
constraint. We give color $1$ to the vertex of $S_j$ representing $a$, and use
colors $2,3$ to color the rest of $S_j$. The only implication gadgets activated
in this way are those incident on the vertex representing $a$; we give to their
other endpoint, which is found in $U_i^{j,a}$, the only viable color. For all
other $a'\neq a$ we give to vertices of $U_i^{j,a'}$ a color that we used in
$U_i^{j,a}$. We claim that this is a proper coloring because the colors we used
in $U_i^{j,a}$ are $\{1,\ldots,k\}\setminus T(v_i)$, while the colors we used
in $V_i$ are $T(v_i)$, hence all edges between $V_i$ and $U_i^{j,a}$ are
properly colored for any $a$.  \end{proof}

\begin{lemma}\label{lem:mpw2}

	If  $G(\phi)$ admits a proper list coloring, then $\phi$ is a
	satisfiable \qcsp instance.

\end{lemma}

\begin{proof}

Suppose that we are given a proper list coloring
$\mathbf{c}:V\to\{1,\ldots,k\}$ for $G$. We extract an assignment for $\phi$ as
follows: for each $i\in\{1,\ldots,n\}$, let $\mathbf{c}(V_i)$ be the set of
colors the coloring uses for $V_i$. Since $V_i$ is a clique, this set includes
$\lfloor k/2\rfloor$ elements. We therefore set $x_i = T^{-1}(\mathbf{c}(V_i))$
and this is well-defined since $T$ is a one-to-one correspondence between
$\{1,\ldots,B\}$ and subsets of $\{1,\ldots,k\}$ of size $\lfloor k/2\rfloor$. 

We argue that this assignment is satisfying. Suppose to the contrary that it
does not satisfy a clause $c_j$. Because of the OR gadget, one of the vertices
of $S_j$ has received color $1$, say the vertex that represents the satisfying
assignment $a$ of $c_j$. This assignment must disagree with our assignment in
some variable that appears in $c_j$, say the variable $x_i$. Suppose that $a$
assigns value $v'$ to $x_i$, while our assignment has given value $v$ to $x_i$.

Because of the implication gadgets incident on the vertex representing $a$ in
$S_j$, we have that $U_i^{j,a}$ uses the $\lceil k/2\rceil$ colors of
$\{1,\ldots,k\}\setminus T(v')$. However, $V_i$ uses the $\lfloor k/2\rfloor$
colors of $T(v)$. If $T(v)\neq T(v')$ then $T(v) \cap (\{1,\ldots,k\}\setminus
T(v')) \neq \emptyset$, which contradicts the correctness of the coloring.
\end{proof}

\begin{lemma}\label{lem:mpw3}

	$G(\phi)$ can be constructed in time polynomial in $|\phi|$, and
$\mpw(G) \le n + O(1)$.

\end{lemma}

\begin{proof}

We first observe that, for all $i\in\{1,\ldots,n\}$, the vertices of $V_i$ are
true twins with the same list ($\{1,\ldots,k\}$). For the purposes of computing
the modular pathwidth of the graph, we can therefore retain a single vertex of
each $V_i$. We now delete these $n$ vertices, and what remains is to show that
the graph we obtain has pathwidth $O(1)$.

By Lemma \ref{lem:weak-pw} we can replace all implication gadgets by simple
edges, and this will not decrease the pathwidth of the graph by more than a
small constant. We observe that after removing the sets $V_i$, the graph is
disconnected, and we have one component for each $j\in\{0,\ldots,m\}$. This
component contains an OR($S_j$) gadget, and the vertices of the sets
$U_i^{j,a}$. However, all such vertices are now leaves, because each such
vertex is connected to a unique vertex of $S_j$ and its other neighbors (the
set $V_i$) no longer exist in the graph. If we remove such leaves the graph
that remains is simply the OR($S_j$) gadgets, which form paths. Hence, the
graph obtained after removing the sets $V_i$ has constant pathwidth.
\end{proof}

\begin{proof}[Theorem \ref{thm:mpw}]

	The proof now follows from the described construction, Theorem
\ref{thm:csp}, and Lemmata \ref{lem:list}, \ref{lem:mpw1}, \ref{lem:mpw2},
\ref{lem:mpw3} in a same way as the proof of Theorem \ref{thm:cw}. \end{proof}

\subsection{ETH-based Lower Bound}

\begin{theorem}\label{thm:mpw-eth}

If there exists an algorithm that solves $k$-\kc on instances with $n$ vertices
and $k=O(\log n)$ in time $n^{o(\mpw)}$, then the ETH is false.  As a result,
if there is an algorithm solving $k$-\kc in time $2^{o(k\cdot\mpw)}$, then the
ETH is false.

\end{theorem}

\begin{proof}[Theorem \ref{thm:mpw-eth}]

The construction we use to prove Theorem \ref{thm:mpw-eth} is identical to the
construction presented in the previous section. We change, however, our
starting point: given an $n$-variable \tsat instance, we will construct a \qcsp
instance where $q=3$ and $B=n$. 

More specifically, given an instance $\phi_1$ of \tsat with $n$ variables, our
first step will be to construct an instance $\phi_2$ of \qcsp.  We assume
without loss of generality that $n$ is a power of $2$. We partition the
variables of $\phi$ into $t$ sets $V_1,\ldots,V_t$, such that each set contains
at most $\log n$ variables. Therefore $t=\lceil n/\log n/\rceil$. For each
group of variables $V_i$ we define a \csp variable $x_i$ that will take values
in $\{1,\ldots,n\}$. We make a one-to-one correspondence translation between
values of $x_i$ and truth assignments for the variables of $V_i$. We now
construct the constraints of $\phi_2$ as follows: for each clause $c_j$ of
$\phi_1$, suppose that $c_j$ involves three variables, from the groups
$V_{j_1},V_{j_2},V_{j_3}$.  We construct a constraint $c_j'$ involving the
variables $x_{j_1},x_{j_2},x_{j_3}$. The satisfying assignments of this
constraint are all assignments whose translation satisfies $c_j$. It is not
hard to see that the construction of $\phi_2$ can be done in polynomial time,
and that we produce an equivalent instance.

We have now constructed a \qcsp with $N=\lceil n/\log n\rceil$ variables,
alphabet size $B=n$, and arity $q=3$. If there exists an algorithm solving this
instance in $B^{o(N)} = 2^{o(n)}$, then the ETH is false. Set $k:=2\log n$. We
observe that ${k\choose \lfloor k/2\rfloor} = \frac{(2\log n)!}{(\log n)!(\log
n)!} \ge 2^{\log n} = B$. We therefore perform the construction of Theorem
\ref{thm:mpw}, where we have $q=3$, $k=2\log n$, and we define the translation
function $T$ so that it is one-to-one; this is possible since ${k\choose k/2}
\ge B$. All the arguments of the construction go through unchanged, and the
produced graph has size polynomial in $n$ and modular pathwidth $N+O(1) =
\frac{n}{\log n}+O(1)$.  Therefore, if there exists an algorithm running in
time $|G|^{o(\cw)} = 2^{o(n)}$ this would contradict the ETH. Similarly, if
there was an algorithm running in $2^{o(k\cdot\mpw)} = 2^{o(n)}$ this would
contradict the ETH.  \end{proof}

%% file: algs.tex
We present two algorithms establishing that the lower bounds of Sections
\ref{sec:cw},\ref{sec:mtw} are essentially tight.  Though both algorithms are
based on standard techniques, we remark that the algorithm for clique-width
requires some extra effort to obtain a DP of the promised size.

\subsection{Clique-width algorithm}

Our algorithm is based on standard DP. Its basic idea is that a partial
solution is characterized by the set of colors it uses to color a set of
vertices that share the same label. This leads to a DP table of size
$(2^k-1)^{\cw}$, by observing that for any non-empty label set, any viable
partial solution will use at least one color, hence there are $2^k-1$ possible
subsets of $\{1,\ldots,k\}$ to consider. To improve this to $(2^k-2)^{\cw}$,
which would match the lower bound of Theorem \ref{thm:cw} we need a further
idea which will allow us to also rule out the set that uses all $k$ colors.

Let $t$ be a node of the binary tree representing the clique-width expression
of $G$, and let $V_t^i$ be the set of vertices that have label
$i\in\{1,\ldots,\cw\}$ in the labeled graph $G_t$ produced by the
sub-expression rooted at $t$. We will say that $V_t^i$ is a \emph{live} label
set if there exists an edge in $G$ that is incident on a vertex of $V_t^i$ and
does not appear in $G_t$. In other words, a label set is live if there is a
join operation that involves its vertices which has not yet appeared in $t$.
The main observation is that live label sets cannot use all $k$ colors in a
valid partial solution, since then the subsequent join operation will fail.
Non-live label sets, on the other hand, are irrelevant, since if the coloring
is already valid for such a set it is guaranteed to remain valid. Our DP
algorithm will therefore keep track of the partial colorings only of live label
sets, and thus produce a DP table of size $(2^k-2)^{\cw}$.  In this sense, our
DP algorithm is slightly non-standard, as part of its procedure involves
``looking ahead'' in the graph to determine if a label set is live or not.
What remains is the problem of implementing the DP so that it takes time linear
in the table size;  this is handled using the techniques introduced in
\cite{BodlaenderLRV10,RooijBR09}.

\begin{theorem}\label{thm:cw-alg}

There is an algorithm which, given a graph $G$, an integer $k$, and a
clique-width expression for $G$ with $\cw$ labels decides if $G$ is
$k$-colorable in time $O^*\left((2^k-2)^{\cw}\right)$.

\end{theorem}

\begin{proof}

We assume we are given a binary tree representing the clique-width expression
that produces $G$. For each node $t$ let $\ell(t)$ be the number of live labels
in the graph $G_t$, that is, the graph produced by the sub-expression rooted at
$t$. A label $i$ is live at $G_t$ if the set $V_t^i$ of vertices that have this
label in $G_t$ is non-empty, and there is an edge in $G$ that is incident on a
vertex of $V_t^i$ and does not appear in $G_t$. Clearly, $\ell(t)\le \cw$.

Let $\mathcal{C} := 2^{\{1,\ldots,k\}} \setminus \{\emptyset, \{1,\ldots,k\}\}$
be the set of all interesting sets of colors, that is, all sets of colors
except the empty set and the set that contains all colors.  For each node $t$
of the tree we will maintain a dynamic programming table $A_t :
\mathcal{C}^{\ell(t)} \to \mathbb{N}$.  The meaning of such a table is the
following: its input, describes the signature of a partial solution, that is, a
valid $k$-coloring of $G_t$, in the sense that it defines for each live label
$i$ the set of colors used in a partial coloring to color the vertices of
$V_t^i$.  Given such a signature the table stores \emph{the number} of
colorings of $G_t$ that agree with this signature.  The reason we select to
solve the more general counting problem is that this allows us to adopt the
techniques of \cite{RooijBR09} to speed up that computation of the table in
union nodes. To ease notation, we will interpret signatures $S\in
\mathcal{C}^{\ell(t)}$ as functions which, given a live label return the set of
colors used in this label set in signature $S$.

Before proceeding, let us explain why this table, whose size is
$O^*\left((2^k-2)^{\cw}\right)$ stores sufficient information to solve the
problem. First, it is not hard to see that it is not a problem that $A_t$ does
not store information regarding empty label sets, hence it is not a problem
that $\emptyset\not\in\mathcal{C}$. Second, if a partial solution uses all
colors $\{1,\ldots,k\}$ to color a live label set $V_t^i$, it is clear that
this solution cannot be extended to a coloring of the whole graph, as the
vertices of $V_t^i$ will later acquire a new common neighbor (since $i$ is
live), for which no color is available. Hence, not considering such solutions
for live labels is also not a problem. Finally, if a partial solution
represents valid $k$-colorings of $G_t$, vertices belonging to a non-live label
have all their edges properly colored. Since the label is non-live, these
vertices already have their incident edges properly colored in $G$, $A_t$ does
not need to store information on non-live labels, apart from the fact that the
coloring of $G_t$ is valid.

Given the above definition of the table we will now proceed inductively on its
calculation. For introduce nodes that construct a vertex with label
$l\in\{1,\ldots,\cw\}$ there is only one live label, so for each color
$c\in\{1,\ldots,k\}$ the table $A_t$ has an entry $A_t[ \{c\} ] = 1$. 

For a join node $t$ with child node $t'$ where the join operation is performed
between labels $i_1,i_2$, for every signature $S \in \mathcal{C}^{\ell(t')}$
such that $S(i_1)\cap S(i_2) \neq \emptyset$ we update the corresponding entry
of $A_{t'}$ to $0$. Then, if $i_1,i_2$ are still live labels in $t$ we simply
copy $A_{t'}$ to $A_t$. If $i_1$ has become non-live in $t$ we set $A_t[S] =
\sum_{S_1\in\mathcal{C}} A_{t'}[S\times S_1]$, that is, to compute the number
of solutions that produce a signature $S$ in $t$ we consider all signatures
$S'$ in $t'$ which extend $S$ by giving a set of colors to label $i_1$, and
take the sum of the corresponding entries. We do a similar operation if $i_2$,
or both $i_1,i_2$ become non-live.

For a rename node $t$ with child node $t'$ where the rename operation is
performed from label $i_1$ to label $i_2$, we have several cases. First, if
$i_1$ was a non-live label in $t'$, then $i_2$ is also non-live in both $t,t'$.
In this case, we simply copy the table $A_{t'}$ to $A_t$. Second, if $i_1$ is
live in $t'$, then $i_2$ is live in $t$. Now, suppose that $i_2$ is non-live in
$t'$. This means that $i_2$ is empty in $t'$, so we can again copy $A_{t'}$ to
$A_{t}$ replacing label $i_1$ with $i_2$ in every signature. Finally, the
interesting case is when $i_1,i_2$ are both live in $t'$, therefore $i_2$ is
live in $t$.  We have $\ell(t) = \ell(t')-1$. We initialize $A_t$ to be $0$
everywhere and then, for every entry of $A_{t'}$ we do the following: if
$A_{t'}[S] = v$ we take the signature $S'$ in $t$ which sets $S'(i_2) =
S(i_1)\cup S(i_2)$, and $S'(i')=S(i')$ for $i'\neq i_2$, and add to $A[S']$ the
value $v$, under the condition that $S(i_1)\cup S(i_2)\neq \{1,\ldots,k\}$.  In
other words, in order to calculate how many solutions give a set of colors to
label $i_2$ in $t$, we consider all solutions in $t'$ such that the union of
colors used in $i_1,i_2$ gives this set. As explained, we disregard signatures
$S$ for which $S(i_1)\cup S(i_2)=\{1,\ldots,k\}$, because such signatures do
not represent solutions that can be extended to colorings of $G$, due to the
liveness of $i_1,i_2$.

Finally, the most challenging part of this algorithm is how to handle a union
node $t$ with two children $t_1,t_2$. Here, we will need to transform the
tables $A_{t_1},A_{t_2}$ in a form that will more easily allow us to combine
the solutions, along the lines of the technique used in \cite{RooijBR09}. In
particular, define two tables $B_{t_1},B_{t_2}: \mathcal{C}^{\ell(t)} \to
\mathbb{N}$ whose intended meaning is the following: the input is a signature
$S$ giving set of colors for every live label of $t$ (note that some of these
labels may be empty in $t_1$ or $t_2$); a partial solution in $G_{t_1}$
satisfies this signature if for each live label $i$ of $t$ the coloring uses
\emph{a subset of} $S(i)$ to color $V_{t_1}^i$ (and similarly for $G_{t_2}$).
In other words, the difference between tables $A,B$ is that in $A$ the sets of
colors given in the signature are exact, while in $B$ they form upper bounds on
the set of allowed colors that a partial solution may use in a label set. It is
not hard to see that the size of $B_{t_1},B_{t_2}$ is
$O^*\left((2^k-2)^{\cw}\right)$, because once again we are only interested in
live labels of $t$.

We now need to describe three steps: first, how we can convert
$A_{t_1},A_{t_2}$ to $B_{t_1},B_{t_2}$ in time almost-linear in the size of the
tables; second, how we can compute $B_t$ from $B_{t_1},B_{t_2}$ in the same
time; third, how we can obtain $A_t$ from $B_t$.

For the first part, we proceed inductively. Let $\ell = \ell(t_1)$ be the
number of live labels of $t_1$, and suppose without loss of generality that
these labels are $\{1,\ldots,\ell\}$.  For each $i\in \{1,\ldots, \ell\}$ and
$j\in\{0,1,\ldots,k+1\}$ we define $B^{i,j}_{t_1}[S]$ as the number of
colorings with the following property: if the coloring uses the set of colors
$C_l$ to color $V_{t_1}^l$ then for all $l< i$, $C_l\subseteq S(l)$; for all
$l>i$, $C_l=S(l)$; and for $l=i$, $C_l\cap\{1,\ldots,j\}\subseteq
S(l)\cap\{1,\ldots,j\}$, and $C_l\cap\{j+1,\ldots,k\} =
S(l)\cap\{j+1,\ldots,k\}$. We observe that $B_{t_1}^{1,0} = A_{t_1}$, while
$B_{t_1}^{l,k+1}=B_{t_1}$. We also observe that $B_{t_1}^{i,k+1} =
B_{t_1}^{i+1,0}$, for all $i\in\{1,\ldots,\ell-1\}$. Therefore, it suffices to
explain how to compute $B_{t_1}^{i,j}$ from $B_{t_1}^{i,j-1}$ for all
$i\in\{1,\ldots,\ell\}$ and $j\in\{1,\ldots,k+1\}$ in time
$O^*\left((2^k-2)^{\cw}\right)$, and then we can repeat this process $k\ell$
times. We now use the fact that for all signatures $S$, such that $j\not\in
S(i)$, $B_{t_1}^{i,j}[S] = B_{t_1}^{i,j-1}[S]$; while for signatures $S$ where
$j\in S(i)$ we have $B_{t_1}^{i,j}[S] = B_{t_1}^{i,j-1}[S] +
B_{t_1}^{i,j-1}[S']$, where $S'$ is the signature that agrees with $S$
everywhere, except it sets $S'(i) = S(i)\setminus\{j\}$. The correctness of
this procedure follows directly from the definitions given above.  Note that we
have defined $B_{t_1}$ on signatures on live labels in $t_1$, rather than
signatures on live labels of $t$. However, the labels which are live in $t$ but
not $t_1$ are empty in $G_{t_1}$, hence their coloring is irrelevant, and a
version of $B_{t_1}$ that uses signatures on live labels of $t$ can easily be
obtained. We obtain $B_{t_2}$ with the same algorithm.

For the second part, we observe that for all signatures $S$ we have $B_t[S] =
B_{t_1}[S] \times B_{t_2}[S]$. Since multiplication of numbers with $n^{O(1)}$
bits can be dones in $n^{O(1)}$ time, this step also takes time
$O^*\left((2^k-2)^{\cw}\right)$.

Finally, to convert $B_t$ to $A_t$ we use the reverse of the procedure we
described. We again define $B_t^{i,j}$ in the same way, and compute
$B_t^{i,j-1}$ from $B_t^{i,j}$ for all $i\in\{1,\ldots,\ell\}$,
$j\in\{1,\ldots,k+1\}$. For a signature $S$ such that $j\not\in S(i)$ we set
$B_t^{i,j-1}[S] = B_t^{i,j}[S]$; while if $j\in S(i)$, and $S'$ is the
signature that agrees with $S$ everywhere except $S'(i)=S(i)\setminus\{j\}$, we
set $B_t^{i,j-1}[S] = B_t^{i,j}[S] - B_t^{i,j}[S']$.  \end{proof}

\subsection{Modular Treewidth Algorithm}

For modular treewidth, we remark that $k$-\kc for this parameter can be seen as
an equivalent version of \textsc{Multi-Coloring} parameterized by treewidth.
In \textsc{Multi-Coloring}, each vertex $v$ has a demand $b(v)$, and we are
asked to assign $b(v)$ distinct colors to each vertex so that neighboring
vertices have disjoint colors (see e.g.  \cite{BonamyKPSW17}).  In our context,
the vertex representing a class of $b$ true twins corresponds to a vertex with
demand $b$.

\begin{theorem}\label{thm:mtw-alg}

There is an algorithm which, given a graph $G$, an integer $k$, and a modular
	tree decomposition of $G$ of width $\mtw$, decides if $G$ is
	$k$-colorable in time $O^*\left({k\choose \lfloor
	k/2\rfloor}^{\mtw}\right)$.

\end{theorem}

\begin{proof}

	The algorithm uses standard DP techniques, so we will sketch some of
	the details. We assume that we are given the graph $G^t$, where each
	twin class is represented by a single vertex (otherwise, $G^t$ can be
	computed in polynomial time). Furthermore, we can assume that $G$ did
	not contain any false twins, because if for two vertices $u,v$ we have
	$N(u)=N(v)$ then deleting one of the two vertices does not affect the
	$k$-colorability of the graph. Hence, every vertex of $G^t$ represents
	a class of true twins in $G$. For each vertex $v$ of $G^t$ we define
	$b(v)$ as the size of the class of true twins that contains $v$ in $G$
	(therefore, for vertices without twins in $G$ we have $b(v)=1$).  We
	observe that $\max_{v\in V} b(v) \le k$, because otherwise there exists
	a clique of size $k+1$ and we can immediately reject.

	We will now solve \textsc{Multi-Coloring} on $G^t$ with the demands
	$b(v)$ we defined. We recall that the goal is to assign to each vertex
	$v$ a subset $C(v)\subseteq \{1,\ldots,k\}$ with $|C(v)|=b(v)$ so that
	for any edge $(u,v)\in E$ we have $C(v)\cap C(u) = \emptyset$.

	Our algorithm will, in every bag consider all possible assignments of
	$b(v)$ colors to each vertex $v$ of the bag. We reject partial
	solutions for which two neighbors in the bag have non-disjoint
	assignments. It is not hard to see how this table can be maintained
	with standard DP techniques in time linear in its size, if we are
	working with a nice tree decomposition \cite{BodlaenderK08}. Let us
	therefore bound the size of the table. For a vertex $v$ with demand
	$b(v)$ our algorithm will consider all $k\choose b(v)$ possible
	colorings. However, ${k\choose b(v)} \le {k\choose \lfloor
	k/2\rfloor}$, we therefore have that the size of the table is at most
	${k\choose \lfloor k/2 \rfloor}^{\mtw}$.  \end{proof}